\documentclass[reqno,12pt,letterpaper]{amsart}
\usepackage{amsmath,amssymb,amsthm,graphicx,mathrsfs,url}
\usepackage[usenames,dvipsnames]{color}
\usepackage[colorlinks=true,linkcolor=Red,citecolor=Green]{hyperref}
\usepackage{amsxtra}
\usepackage{amscd}
\usepackage[normalem]{ulem}
\usepackage{wasysym} % for "\varhexagon" macro
\usepackage{graphicx}% for "\scalebox" macro
\usepackage{subcaption}
\def\arXiv#1{\href{http://arxiv.org/abs/#1}{arXiv:#1}}
\usepackage{array}
\newcolumntype{P}[1]{>{\centering\arraybackslash}m{#1}}

\def\?[#1]{\textbf{[#1]}\marginpar{\Large{\textbf{??}}}}

\def\smallsection#1{\smallskip\noindent\textbf{#1}.}
\let\epsilon=\varepsilon % sorry Knuth
\newcommand{\CC}{{\mathbb C}}

\newcommand{\ZZ}{{\mathbb Z}}

\newtheorem{theo}{Theorem}
\newtheorem{prop}{Proposition}[section]

\newtheorem{lemm}[prop]{Lemma}

\numberwithin{equation}{section}

\DeclareMathOperator{\Spec}{Spec}

\let\Im=\Imag

\let\Re=\Real

\usepackage{scalerel}

\newcommand\reallywidehat[1]{\arraycolsep=0pt\relax%
\begin{array}{c}
\stretchto{
  \scaleto{
    \scalerel*[\widthof{\ensuremath{#1}}]{\kern-.5pt\bigwedge\kern-.5pt}
    {\rule[-\textheight/2]{1ex}{\textheight}} %WIDTH-LIMITED BIG WEDGE
  }{\textheight} % 
}{0.5ex}\\           % THIS SQUEEZES THE WEDGE TO 0.5ex HEIGHT
#1\\                 % THIS STACKS THE WEDGE ATOP THE ARGUMENT
\rule{-1ex}{0ex}
\end{array}
}

\def\blue#1{\textcolor{blue}{#1}}
\def\red#1{\textcolor{red}{#1}}

\title{Fine structure of flat bands in a chiral model of magic angles}

\author{Simon Becker}
\email{simon.becker@math.ethz.ch}
\address{ETH Zurich, 
Institute for Mathematical Research, 8092 Zurich, CH.}

\author{Tristan Humbert}
\email{tristan.humbert@ens.psl.eu}
\address{Department of Mathematics, University of California,
Berkeley, CA 94720, USA.}

\author{Maciej Zworski}
\email{zworski@math.berkeley.edu}
\address{Department of Mathematics, University of California,
Berkeley, CA 94720, USA.}

\begin{document}
\begin{abstract}
We analyze symmetries of Bloch eigenfunctions at magic
angles for the Tarnopolsky--Kruchkov--Vishwanath chiral model of the
twisted bilayer graphene (TBG) following the framework introduced by
Becker--Embree--Wittsten--Zworski. We show that
vanishing of the first Bloch eigenvalue away from the Dirac points
implies its vanishing at all momenta, that is the existence of a flat
band. We also show how the multiplicity of the flat band is related to the nodal set of the Bloch eigenfunctions.
%We also demonstrate that for a generic choice of tunneling potentials, obeying
%all translational and rotational symmetries, the Hamiltonian only
%exhibits flat bands of minimal multiplicity. 
We conclude with two
numerical observations about the structure of flat bands.

\end{abstract}

\maketitle 

\section{Introduction}

In this article we study the chiral version \cite{sgog,magic} of the Bistritzer--MacDonald 
Hamiltonian \cite{BM11} describing twisted bilayer graphene:
\begin{equation}
\label{eq:defD} 
H ( \alpha ) := \begin{pmatrix} 0 & D ( \alpha )^* \\
D ( \alpha ) & 0 \end{pmatrix} , \ \ \ 
D ( \alpha ) := \begin{pmatrix}
2 D_{\bar z } & \alpha U ( z ) \\
\alpha U ( - z )& 2 D_{\bar z } \end{pmatrix},  
\end{equation}
where $ U $ is a real analytic function on $ \mathbb C = \mathbb R^2 $, 
and 
\begin{equation}
\label{eq:propU}
\begin{gathered} 
  U ( z + \gamma ) = e^{ i \langle \gamma , K  \rangle } U ( z ) , \ \ U ( \omega z ) = \omega U(z) , \ \ \overline {U ( \bar z ) } = - U ( - z ) , \ \ \ 
\omega = e^{ 2 \pi i/3},  \ \\ \gamma \in \Lambda := \omega \mathbb Z \oplus \mathbb Z , \ \ 
\omega K \equiv K \not \equiv 0 \!\!\! \mod \Lambda^* , \ \ 
\Lambda^* :=  \frac {4 \pi i}  {\sqrt 3}  \Lambda , \ \ \langle z , w \rangle := \Re ( z \bar w ) .
\end{gathered}
\end{equation}

The most studied case is  the Bistritzer--MacDonald potential 
which in the convention of \eqref{eq:propU} corresponds to
\begin{equation}
\label{eq:BMU}
U ( z ) =  - \tfrac{4} 3 \pi i \sum_{ \ell = 0 }^2 \omega^\ell e^{ i \langle z , \omega^\ell K \rangle }, \ \ \ K = \tfrac43 \pi  ,
\end{equation}
see the Appendix for the translation of the conventions.

\medskip

\noindent
{\bf Definition.} \emph{A value of $ \alpha $ is called {\em magical} if 
the Hamiltonian $ H ( \alpha ) $ has a {\em flat band} at zero energy $($see \eqref{eq:eigs} below$)$.
This is equivalent to $ \Spec_{ L^2 ( \mathbb C/3 \Lambda ) } D ( \alpha ) 
= \mathbb C $.}

\medskip
In the physics literature -- see \cite{magic} -- $ \alpha $ is a dimensionless parameter
which, modulo physical constants, is proportional to the angle of twisting of the two sheets
of graphene. Hence, large $ \alpha$'s correspond to small angles.

We know from \cite{beta} that the set of magic $ \alpha$'s,  $ \mathcal A $, is a discrete
subset of $ \mathbb C $. In \cite{bhz1} we proved that for the potential
\eqref{eq:BMU} $ \mathcal A $ is in fact infinite. Existence and estimates for the
first {\em real} magic $ \alpha $ were obtained by Luskin and Watson \cite{lawa} 
who implemented the method of \cite{magic} with computer assistance (see also Remarks 
following Theorem \ref{theo:Chern}). We also remark that
a rigorous derivation of the full Bistritzer--MacDonald model was provided in \cite{CGG, Wa22}

 Following the physics literature we consider (unlike in 
\cite{beta}) Floquet theory with respect to moir\'e translations: for 
$  u \in L^2_{\rm{loc}}  ( \mathbb C ; \mathbb C^2 ) $ we put 
\begin{equation}
\label{eq:defLag}  
 {\mathscr L}_{\gamma } u  :=  
\begin{pmatrix} e^{ i \langle \gamma, K \rangle }   & 0  \\
0 & e^{-  i \langle \gamma , K \rangle } 
\end{pmatrix}  u ( z + \gamma )   ,   \ \ \ 
\gamma \in \Lambda,  \ \  K = \tfrac43 \pi .  \end{equation}
(Here and elsewhere $ \langle z , w \rangle := \Re z \bar w,$  $ z , w \in \mathbb C$.)
The action is extended diagonally for $ \mathbb C^4 = \mathbb C^2 \times \mathbb C^2 $
and we use the same notation. We then have have 
$ \mathscr L_\gamma D ( \alpha ) = D ( \alpha ) \mathscr L_\gamma  $ and 
$ \mathscr L_\gamma H ( \alpha ) = H ( \alpha ) \mathscr L_\gamma $. 

It is then natural to look at the spectrum of $ H ( \alpha ) $ satisfying the following boundary 
conditions:
\begin{equation}
\begin{gathered} 
\label{eq:FL_ev}  H ( \alpha ) u = E u , \ \ 
u \in H^1_{k} ( \mathbb C/\Lambda, \mathbb C^4 ) ,
\ \ \ H^s_k ( \mathbb C, \mathbb C^4 ) := L^2_k ( \mathbb C ; 
\mathbb C^4 ) 
\cap H^s_{\rm{loc} } ( \mathbb C ; \mathbb C^4 ), 
 \\
L^2_{k} ( \mathbb C/\Lambda, \mathbb C^4 ) := 
\{ u = L^2_{\rm{loc}} ( \mathbb C ; \mathbb C^4 ) : \mathscr L_\gamma u = 
e^{ i \langle k, \gamma \rangle } u \}  . 
\end{gathered}
 \end{equation}
The spectrum is discrete and symmetric with respect to the origin and we index it as follows
(with $ \mathbb Z^* :=  \mathbb Z \setminus \{ 0 \} $)
\begin{equation}
\label{eq:eigs} 
\begin{gathered} \{ E_{ j } ( \alpha, k ) \}_{ j \in  \mathbb Z^* } ,  \ \ \  E_{j } ( \alpha, k ) 
= - E_{-j} ( \alpha , k ) , \\ 0 \leq E_1 ( \alpha, k ) \leq E_2 ( \alpha, k ) \leq \cdots , \ \ \  E_1 ( \alpha, K) =  E_1 ( \alpha, - K ) = 0 , 
\end{gathered} \end{equation}
see \S \ref{s:BF} for more details. The points $ K, -K  $
are called the {\em Dirac points} and are typically denoted by $ K $ and $ K' $ in the physics
literature. (See the appendix to see different $ K $ and $ K'$ when different representation 
of $ \Lambda $ is used.)

\begin{center}
\begin{figure}
\includegraphics[width=11cm]{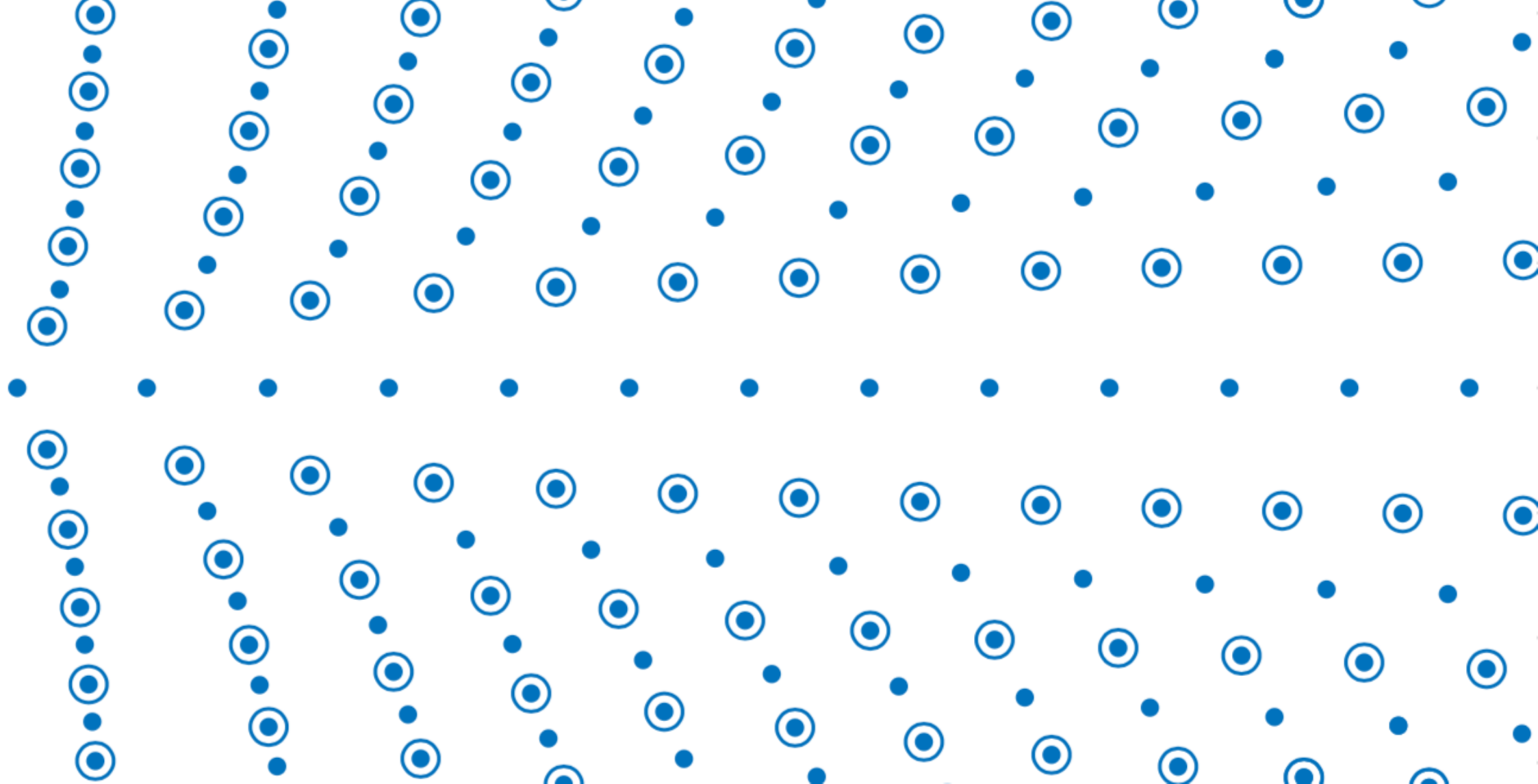}
\caption{\label{f:doub} The multiplicity of the flat band for complex values of $ \alpha $
can be double as illustrated here.
When the potential is 
replaced by $ U_\theta ( z ) = 
\cos \theta U ( z ) +  \sin \theta  \sum_{k=0}^2 \omega^k e^{\bar z \omega^k-  z \bar \omega^k} $ 
%$ U_\theta ( z ) = 
%(\cos^2 \theta) U ( z ) +  ( \sin^2\theta) \sum_{k=0}^2 \omega^k e^{\bar z \omega^k-  z \bar \omega^k} $ 
the symmetries
 \eqref{eq:oldD} (we are using coordinates of \cite{beta} -- see Appendix A) are preserved but the dynamics of $ \alpha$'s is interesting
 when $ \mu $ varies. A movie showing $ \mathcal A $
 as $ \theta $ varies with multiplicities color coded can be found at
  \url{https://math.berkeley.edu/~zworski/multi.mp4}.
 %Kaffeesatzleserei
 }
\end{figure}
\end{center}

The definition of the set of magical $ \alpha $'s can now be rephrased as follows
\begin{equation}
\label{eq:defA}  \mathcal A := \{ \alpha \in \mathbb C : \forall \, k \in \mathbb C, \ \
E_1 ( \alpha, k ) \equiv 0  \} \end{equation}

  Our first theorem states that if the Bloch eigenvalue vanishes 
away from the Dirac points then it vanishes identically, that is the 
band is flat:
\begin{theo}
\label{t:1}
Suppose $ \alpha \in \mathbb C $ and $ E_1 ( \alpha, k ) $ is 
defined using \eqref{eq:FL_ev},\eqref{eq:eigs} for $ H ( \alpha ) $ given 
by \eqref{eq:defD} with $ U $ satisfying \eqref{eq:propU}.  Then 
\begin{equation}  
\label{eq:alphac1}  
\exists \, \ k \notin \{- K, K  \} + \Lambda^*   \ \ \  E_1 ( \alpha, k ) =  0 
\ \Longleftrightarrow \ \forall \, \ k \in \mathbb C  \ \ \  E_1 ( \alpha, k ) =  0 
. \end{equation}
In other words, zero energy band is flat if and only if the Bloch eigenvalue is $ 0 $ at 
some $ k \notin \{ -K, K \} + \Lambda^*  $, 
which is the lattice  of conic points $($see Figure \ref{f:pbands}$)$.
\end{theo}

The next theorem gives a useful criterion for simplicity. It is used in 
\cite{bhz1} to prove existence and  {\em simplicity} of the first magic $ \alpha $ 
and also in \cite{dynip}.

\begin{theo}
\label{t:HtL}
If $ \alpha \in \mathcal A $ 
then, in the notation of \eqref{eq:FL_ev}, 
\begin{equation}
\label{eq:HtL}
\begin{split} 
\forall \, j> 1, \,  k \in \mathbb C  \ \ E_j ( \alpha,  k ) > 0  & \ \Longleftrightarrow \ \forall \,  k \in \mathbb C %\blue{\ k\in  (\tfrac13 \Lambda^*)/\Lambda^*?}: \red{This holds for all k; in an earlier version
%the second line was for $ p \in \tfrac13 \Lambda^*)/\Lambda^* $ but that does not seem to matter?
\ \ 
\dim \ker_{L^2_{k} ( \mathbb C/\Lambda) } D ( \alpha ) = 1 \\
& \ \Longleftrightarrow \ \exists \,  p \in \mathbb C  \ \ 
\dim \ker_{L^2_{p } ( \mathbb C/\Lambda) } D ( \alpha ) = 1. \end{split} 
\end{equation}
\end{theo} 
In other words, the simplicity of
 $ 0 $ as the eigenvalues of 
$ D ( \alpha ) $ on $ L^2_{k } ( \mathbb C /\Gamma; \mathbb C^4 )$ for all $ k$ 
is equivalent to the simplicity of the zero eigenvalue of $ D( \alpha ) $ on
$ L^2_{ p } ( \mathbb C/\Lambda; \mathbb C^2 ) $, for any
 one $ p $.

The symmetries of the potential $ U $ imply that $ U $ vanishes at the {\em stacking point} of high
symmetry: $ z_S := i /\sqrt 3=(\omega-\omega^2)/3\in \Lambda/3$:
\begin{equation}
\label{eq:stacking}  \omega z_S = z_S - 1 - \omega \equiv z_S \mod \Lambda 
\ \Longrightarrow \ U ( - z_S ) = 0 . 
\end{equation}
(To obtain this conclusion use \eqref{eq:propU} to see that $ U ( z_S + \omega \zeta ) = 
\bar \omega U ( z_S + \zeta ) $.)

In the work of Tarnopolsky et al \cite{magic}, flat bands were characterized by vanishing
of a distinguished element of the kernel of $ D( \alpha ) $ at the stacking points $ \pm z_S$. 
For the potential \eqref{eq:BMU} it was claimed that the vanishing of an eigenvector $u\in \mathrm{ker}_{L^2_0}(D(\alpha)-K)$ occurs precisely at $z_S$ . This is equivalent to showing that the zero of $  u \in \ker_{ L^2_{-K} } D ( \alpha ) $ 
occurs precisely at $   z_S$.  We show that 
this is indeed true when $ \alpha \in \mathcal A $ is simple and formulate it more generally: 
%\red{I think at this point it  is not clear what is the link between $\ker_{ L^2_{K} } D ( \alpha )$ and eigenvalues of $D(\alpha)$ for $-K$ on $H_0^1 ( \mathbb C/\Lambda;
%\mathbb C^2 )$. We could refer to proposition \ref{p:prote}, or add in the following theorem that the two functions are linked by $\tau(k)$ and thus share the same zeros . }
%\blue{Actually I got confused now about the convention involving $ u_K $... At this stage I went
%back to $\ker_{L^2_k } $ since that is consistent with \eqref{eq:FL_ev}.}
\begin{theo}
\label{t:zS}
Suppose the equivalent conditions in \eqref{eq:HtL} hold. 
Then, non-trivial elements of $($one dimensional$)$ space 
$ \ker_{ H_k^1 ( \mathbb C/\Lambda;
\mathbb C^2 )} D ( \alpha )  $
 have zeros of order one at 
\begin{equation} 
\label{eq:t3}  \frac{ \sqrt 3 k }{ 4 \pi i } + \Lambda \end{equation} 
and nowhere else. In particular for $ k = - K $ the zeros occur precisely at the 
stacking points $  z_S + \Lambda $. 
\end{theo}

\noindent
{\bf Remark.} We consider the zero at $ z_0 $  to be of order one 
if  $ \partial_z u (z_0) \neq 0 $; the equation implies 
 at zeros $ \partial_{\bar z}^\ell u = 0$ for all $ \ell $ -- see Lemma \ref{l:van}.
This implies that  $ u ( z_0 + \zeta  ) = \zeta w ( \zeta , \bar \zeta ) $,  
$ w ( 0  ) \neq 0 $ and $ w $ is holomorphic near $ 0 \in \mathbb C^2 $. Theorem \ref{t:zS}
is illustrated by Figure~\ref{f:1D}.

As a  consequence we find (in \S \ref{s:ccc}) 
\begin{theo}
\label{theo:Chern}
If  $   \dim \ker_{L^2_{  0} } D( \alpha ) = 1 $, then the Chern number associated with the Bloch function
$ u_{k} \in  \ker_{L^2_{\bf 0} }( D ( \alpha) + k) $, is equal to one $($see \S \ref{s:ccc} for a
precise formulation.$)$ 
\end{theo}

\noindent
{\bf Remarks} 1. Theorem \ref{t:HtL} shows that the assumption of
Theorem \ref{t:zS} %, $ \alpha \in 
%\mathcal A $, $ \dim \ker_{L^2_{  0} } D( \alpha ) = 1 $,   
is equivalent to 
the minimal multiplicity of the flat band, that is to $ |E_j ( \alpha, k) |> 1 $
for $|j | > 1$.

%\noindent
%2. We note that Theorem \ref{t:repr} and \eqref{eq:defup}  then determines zeros
%of other elements of the kernel of $ D ( \alpha ) $ on other elements of $ \Gamma$-periodic functions.

\noindent
2. Numerical results suggest that the first string of complex $ \alpha$'s in $ \mathcal A $ for 
\eqref{eq:BMU} have higher multiplicities 
(see Figure \ref{f:doub} where double $ \alpha$'s are indicated) and in that case the zeros of  $ u_K \in 
\ker_{ H^1_0 } ( D ( \alpha ) + K ) $ appear at $ - z_S + \Lambda $.

\noindent
3. In \cite{bhz1}, we show that the first real angle (existence of which was first 
established by Watson--Luskin \cite{lawa}) is in fact simple. For higher real $ \alpha$'s 
for the potential \eqref{eq:BMU} numerical experiments \cite{beta} provide strong
evidence of simplicity. 

We also make two numerical observations presented in \S \ref{s:chern}.  
The first one is illustrated by Figure \ref{f:pbands} and the movie referenced there.
We see that the rescaled first band is nearly constant close to magic angles
and its shape is closed to that of $ | U | $ after a linear changes of variables
$ z \mapsto k $.

The second observation concerns the behaviour of the curvature of the hermitian 
holomorphic line bundle (somewhat informally) defined by $ k \mapsto u_{k} $
\cite{led} (with hermitian structured inherited from $ L^2 $). We observe that the curvature
peaks at the $ \Gamma $ point, that is, in our notation, at $ k =   i $ -- see Figure
\ref{f:curv}. It is also interesting to note that the curvature does not change much at different
magic $ \alpha$'s -- see \S \ref{s:chern} for definition and computational details.

\begin{center}
\begin{figure}
\includegraphics[width=8cm]{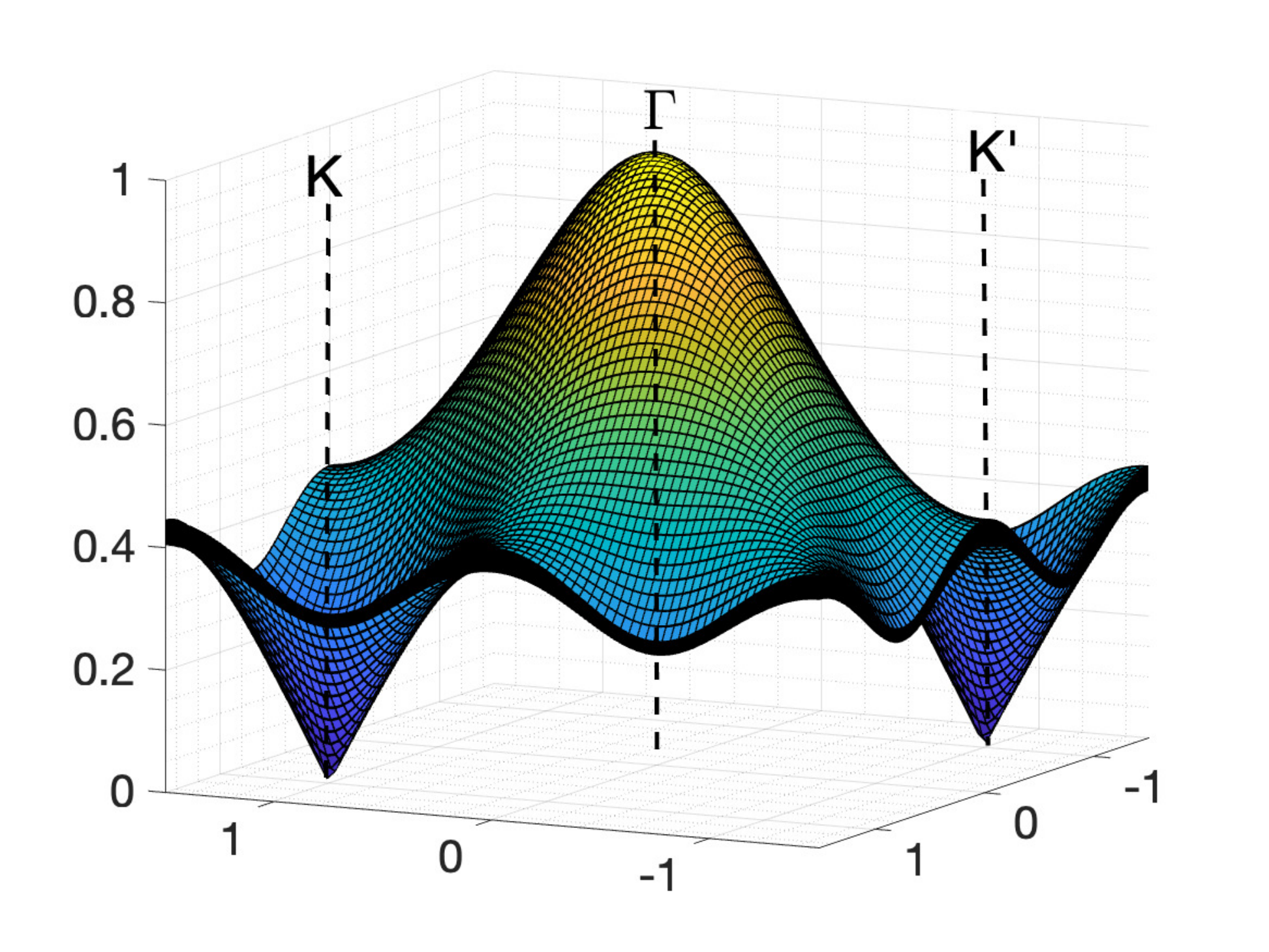}\includegraphics[width=8cm]{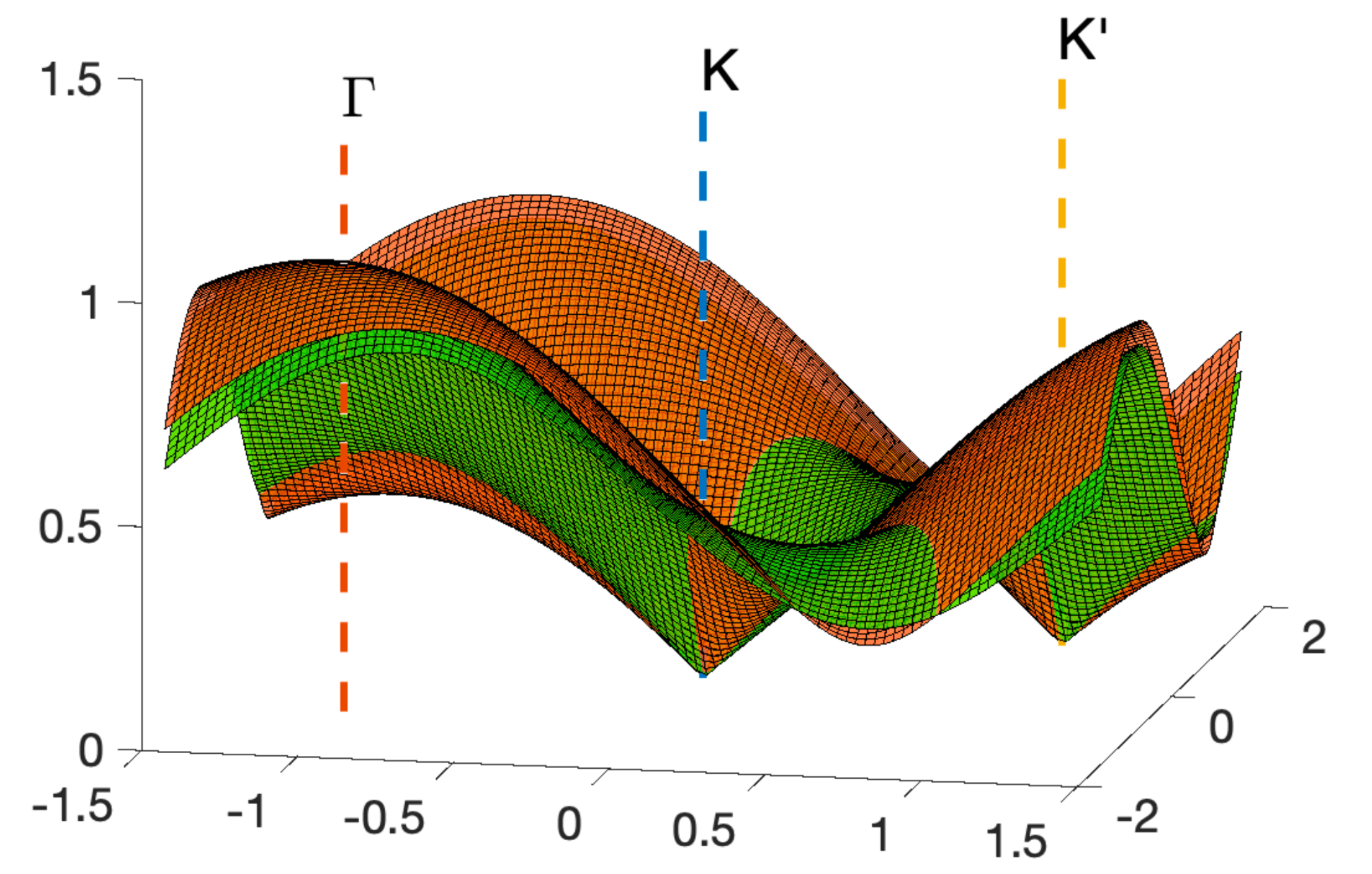}
\caption{\label{f:pbands} Plots of $ k \mapsto E_1 ( \alpha, k )
/(\max_{k} E_1 ( \alpha, k ) ) $ for $  0.4 < \alpha < 0.6 $ (left)
($ k = ( \omega^2 k_1 - \omega k_2 )/\sqrt 3 $, $ |k_j| \leq \frac32 $ and 
we use the coordinates $ k_j $).
Although
the band becomes flat at the first magic $ \alpha \simeq 0.586$, the rescaled plots remain
almost fixed and close to $ {k } \mapsto | 2 \partial_z U ( -4 \sqrt 3 \pi i k/9   ) |$ (right, blue-coloured) compared with $E_1(0.58,k)$ (right, orange-coloured).  For an animated version see 
\url{https://math.berkeley.edu/~zworski/KKmovie.mp4}.}
\end{figure}
\end{center}

\noindent
{\bf Comments on an earlier version of this paper.} We now concentrate exclusively
on the case of simple bands with an expanded discussion of multiplicities moved to \cite{bhz3}.
That paper will also include a modified version of generic multiplicities. Contrary to our
earlier statement, certain double $ \alpha$'s are protected (for instance the ones
marked as double in Figure~\ref{f:doub}).

We conclude this introduction by discussing relation to some physics issues.

\noindent
{\bf The anomalous quantum Hall effect.}
The analysis of the multiplicity of the flat band has immediate implications on the transport properties of twisted bilayer graphene.
% Theorem \ref{t:sim} shows that at a magic angle, the Hamiltonian of twisted bilayer graphene \eqref{eq:defH}  generically exhibits a two-fold degenerate flat band at energy zero. 
In the case of a simple magic angle, the two bands have Chern numbers $\pm 1$ resulting in a net Chern number zero. While this cancellation may sound discouraging at first, it has been recently discovered that twisted bilayer graphene hosts an anomalous quantum Hall effect when it is aligned with hexagonal Boron nitride (hBN) \cite{Serlin}. In that case, an additional sublattice potential of strength $m>0$ is added to the Hamiltonian, that is, the Hamiltonian in \eqref{eq:defD}  is replaced by 
\[ H_{  m}(\alpha) = \begin{pmatrix}m & D(\alpha)^*  \\   D(\alpha) & -m \end{pmatrix}.\]
This effective mass splits the two flat bands at zero energy to one at energy $m$ and one at $-m$, respectively. It follows then from Theorem \ref{theo:Chern} that the anomalous Hall conductivity $\sigma$ of any individual flat band at energy $m$ has Chern number -1 which by the Kubo formula corresponds to a Hall conductivity
\[ \sigma= -\frac{e^2}{2\pi \hbar } c_1.\]
For the band at energy $-m$ the Chern number is $+1.$

\begin{center}
\begin{figure}
\includegraphics[width=10cm]{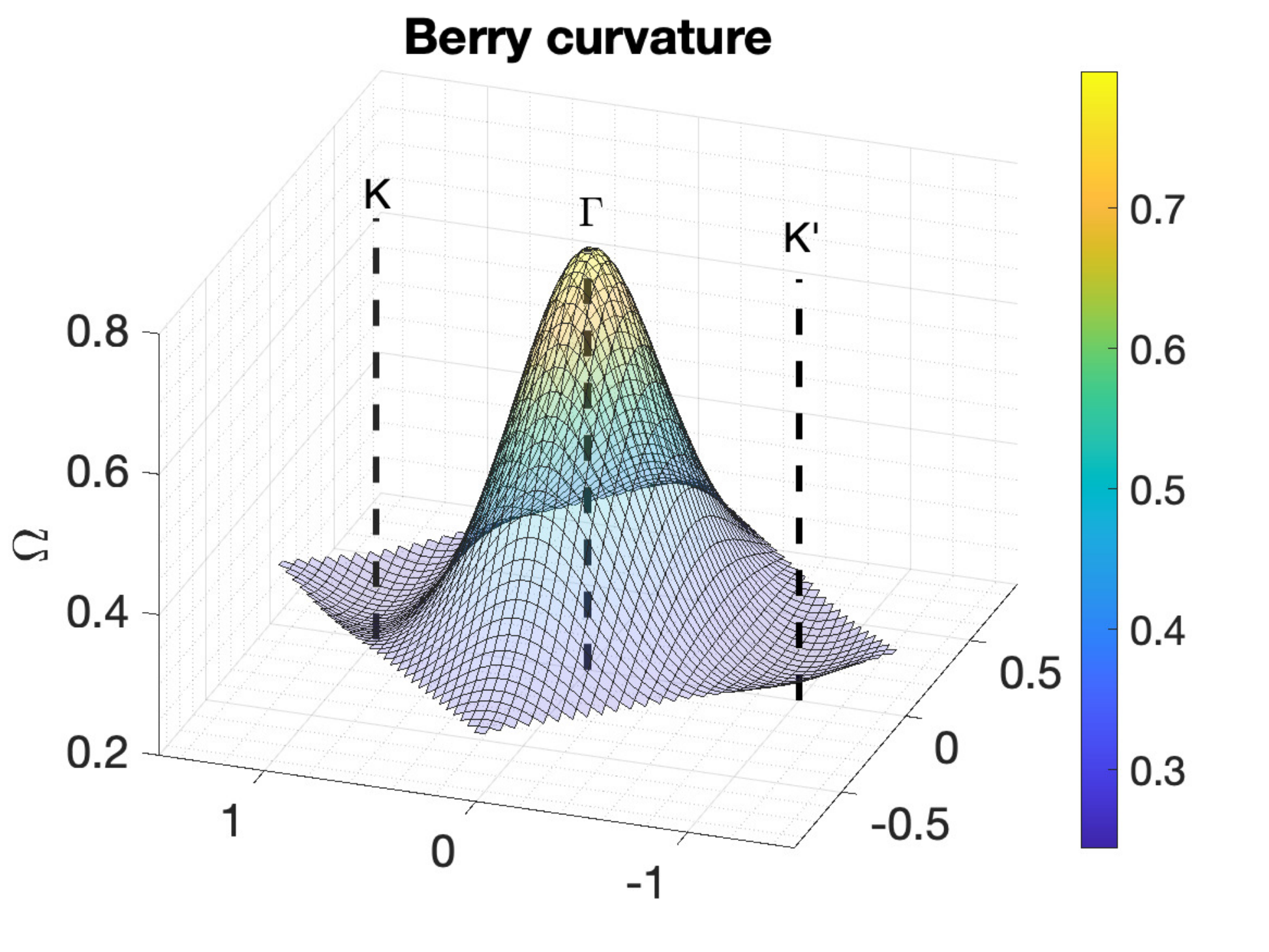}
\caption{\label{f:curv}  The plot of the curvature of the holomorphic line bundle 
corresponding to the first simple band, defined in \eqref{eq:defL}. The extrema at $K,\Gamma,K'$ follow from Prop. \ref{p:H} and the subsequent discussion.}
\end{figure}
\end{center}

\noindent
{\bf Superfluid weights.}
Twisted bilayer graphene exhibits a form of superconductivity at the magic angles. The Bardeen-Cooper-Schrieffer (BCS) theory states \cite{S} that the critical temperature of a superconductor satisfies $T_C \propto e^{-1/(n_F U)}$ where $n_F$ is the DOS at Fermi level and $U$ the interaction between electrons forming a Cooper pair showing why flat bands are promising candidates for high-temperature superconductors. Although this identifies the critical temperature it does not explain whether superconductivity actually exists. Another necessary condition for superconducting states in flat bands has recently been discussed in a series of works by Peotta, Törmä, and collaborators \cite{PT,PT2,PT3,PT4}, see also \cite{PTB} for an analysis in the context of moire materials, stressing the importance of the flat band geometry characterized by the \emph{quantum geometric tensor}. 
The electrodynamic properties of a superconductor are captured by the London equation
$ j=-D_s A$ where $j$ is the current density, $A$ the vector potential in the London gauge and $D_s$ the superfluid weight. In \cite[(22)]{PTB} it is argued that, under some approximations, for a flat band with filling factor $\nu$
\[ D_s \propto \nu(1-\nu) \int_{\text{B.Z.}} g \text{ with metric } g=\partial_{\bar z} \partial_z \log(h(z)) \vert dz \vert^2,\]
that is,  Cooper pairs can support transport in bands with non-trivial topology.
In particular, the volume of the metric is then proportional to the first Chern number \[ D_s \propto \nu(1-\nu) \vert c_1\vert,\]
which by Theorem \ref{theo:Chern} is equal to $1$ for an isolated flat band emphasizing the importance of a non-zero first Chern number in such systems. 

\noindent
{\bf Fractional Quantum Hall effect.}
The multiplicity of the flat band has also implications for other many-body phenomena. Unlike the integer quantum Hall effect which can be understood in a single-particle picture, the fractional quantum Hall effect is a many-body effect conjectured to appear in twisted bilayer graphene \cite{led}. In its original formulation, Laughlin \cite{laugh} constructed under the assumption of a sufficiently large gap of the flat bands, a many-particle wavefunction using the lowest landau levels which was then generalized by Haldane and Rezayi to the torus \cite{HR85}, see also \cite{F15}. Theorem \ref{t:1} together with \cite[Theorem $3$]{bhz1} ensures the existence of such a gap at the first magic angle.
Let us briefly explain the construction in \cite{led}: One defines $\Gamma_{N}:=\frac{4\pi i N_1}{3}\ZZ \omega + \frac{4\pi i N_2}{3}\ZZ \omega^2 $ with $N_s:=N_1N_2$ and $\tau = N_2\omega/N_1.$ If $N_e$ is the number of electrons occupying the band, then we require $m:=N_s/N_e \in 2\mathbb N_0+1.$ 
The ansatz for the Laughlin state of the interacting $N_e$-body electron system, depends on the multiplicity and zero set of the Bloch function, identified in Theorems \ref{t:zS}, 
\[\begin{split} \psi(z_1,...,z_{N_e}) &= F(z_1,...,z_{N_e})  \prod_{i=1}^{N_e} \frac{u(z_i)}{\vartheta_1 \Big(3(z_i+z_S)/(4\pi i \omega) \vert \omega\Big)}\\
F(z_1,...,z_N) &= G_m(Z) \prod_{i<j} g(z_i-z_j).  \end{split}\]
The Bloch conditions $\mathscr L^{(i)}_{N_1a_1+N_2a_2} \psi(z_1,...,z_N) = \psi(z_1,...,z_N),$ where $\mathscr L^{(i)}$ acts like $\mathscr L$ on the $i$-th coordinate $z_i$, is then assumed to hold for each particle and implies that $g$ has a zero of order $m.$ A Laughlin state is then obtained by assuming that all zeros occur at the origin which implies, by assuming $g$ to be holomorphic, that $g(z)= \theta_1 \big(3z/(4\pi i N_1 \omega) \vert \tau\big)^m.$ An easy computation shows that this leaves a $m$-fold degeneracy in the choice of $G$ which is called the \emph{topological order} of the Laughlin state.

\section{Spectral theory and symmetries of the Hamiltonian}
\label{s:specH}

In this section we review symmetries of the Hamiltonian, present a more detailed
discussion of different approaches to Floquet theory, recall the
spectral characterization of magic angles and prove Theorem \ref{t:1}.

\begin{figure}
\includegraphics[height=10cm, width=6cm]{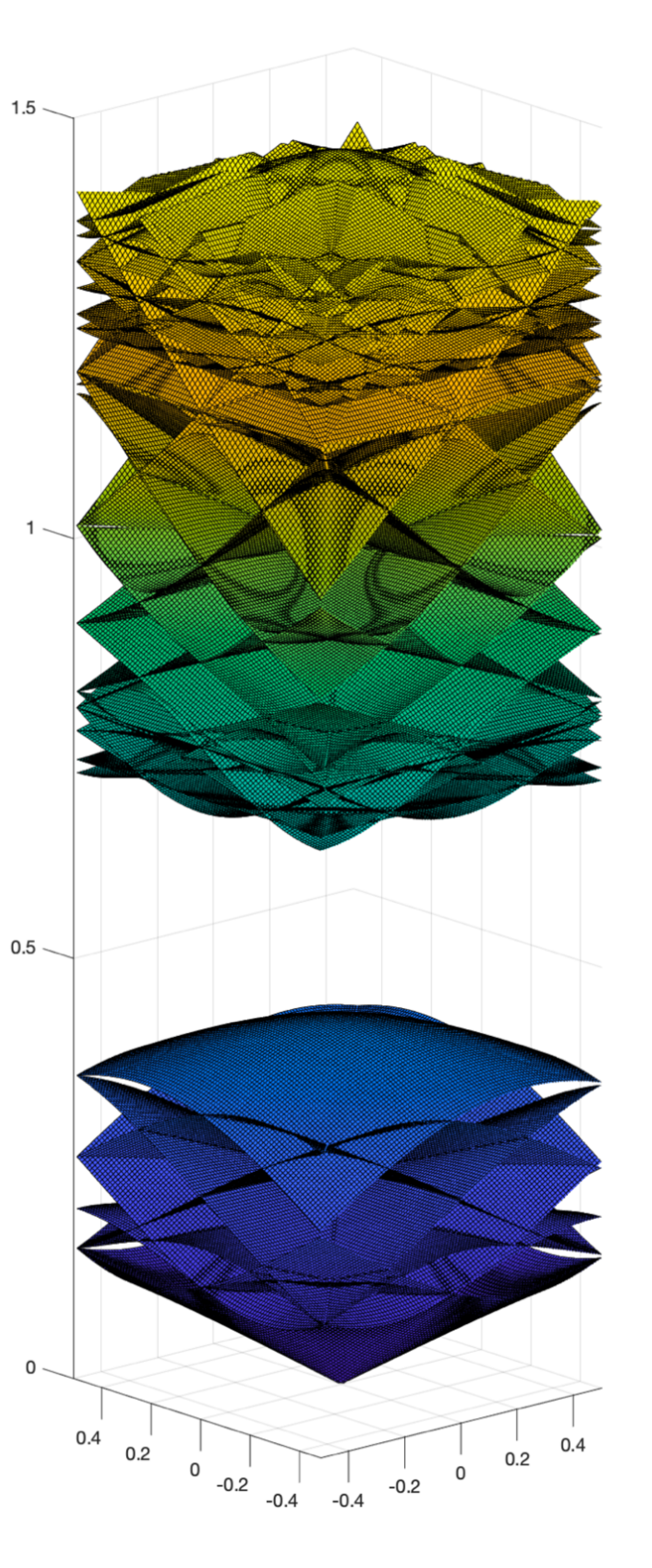}\includegraphics[trim={0cm 1cm 0 0},height=10.5cm, width=7cm]{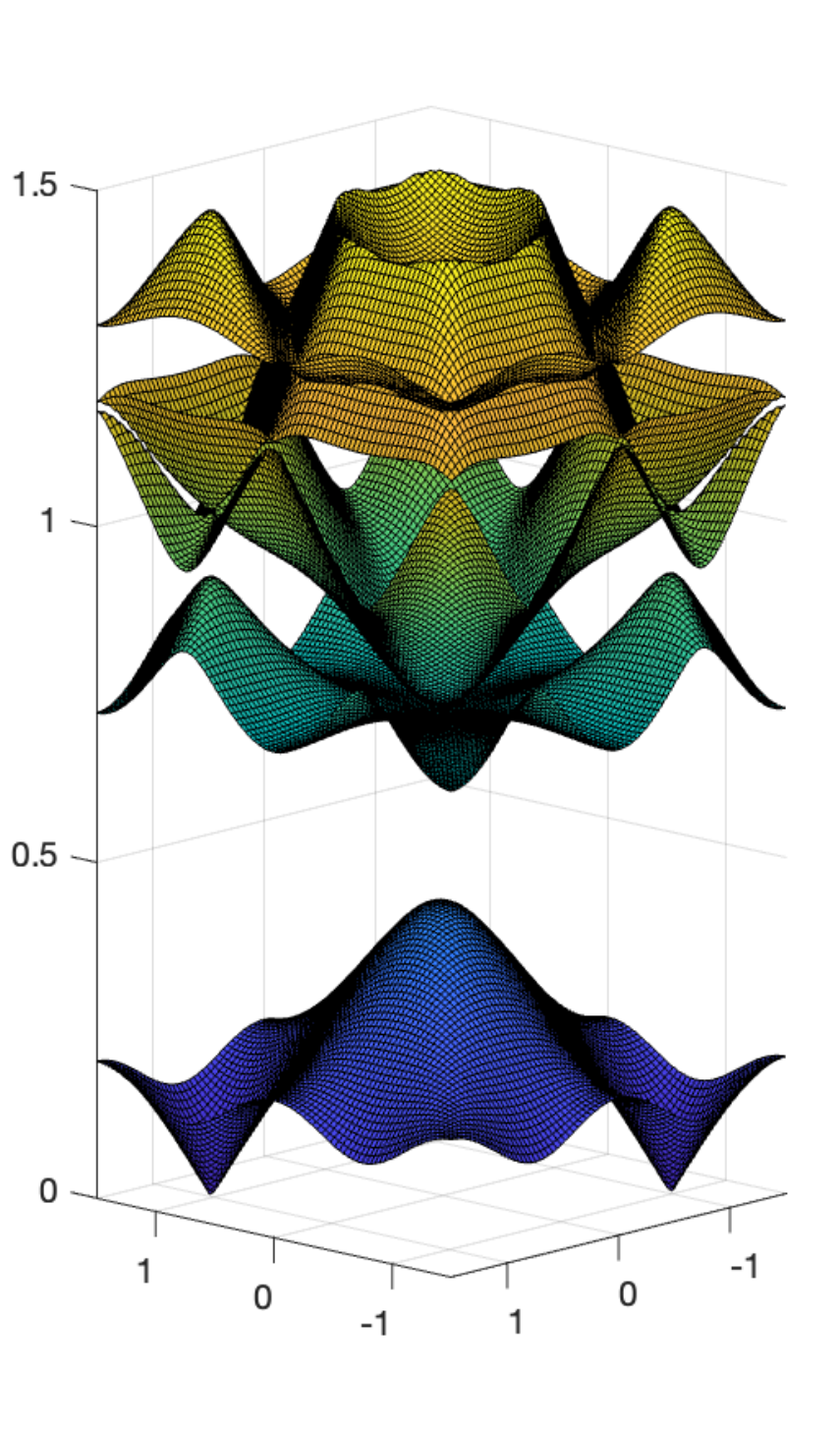}
\caption{\label{f:bands} 
The bands are the functions $ k \mapsto E_j ( \alpha, k )$ where 
$ E_j $ are defined in \eqref{eq:eigs}. 
On the left 
the plot the first 45 bands for $ \alpha = 0.3 $. defined using the boundary condition 
$ u ( z + \gamma ) = e^{i \langle \gamma, k \rangle} u( z ) $, 
$ \gamma \in \Gamma $, $ k \in \mathbb C/ \Gamma^* $,
corresponding to the lattice of exact periodicity of $ D ( \alpha ) $, 
in the convention of \cite{beta} (see Appendix B). 
The fundamental cell of 
$ \Gamma^* $,  parametrized by 
$ ( k_1, k_2 ) \mapsto k = ( \omega^2 k_1 - \omega k_2 ) /\sqrt 3$, 
$ |k_j| < \frac12 $. On the right the plot of $ k \mapsto  
E_j ( 0.3, k ) $, defined using the boundary condition \eqref{eq:FL_ev}
for $ 1 \leq j \leq 5 $, where $ k $ in the fundamental cell of 
$ 3 \Gamma^* $,  parametrized by 
$ ( k_1, k_2 ) $
$ |k_j| < \frac32 $. 
%The band surfaces above a cell of $ 3 \Gamma^* $ 
%on the right are ``chopped up" and moved
%over to a single cell of $ \Gamma^* $ on the left. The structure is much more clearly 
%visible in the picture on the right. 
A movie version of the picture on the right can be found at
\url{https://math.berkeley.edu/~zworski/chiral_bands.mp4}. It is interesting to compare this 
to the case of the full Bistritzer--MacDonald model \cite{BM11} 
\url{https://math.berkeley.edu/~zworski/BM_bands.mp4} where, in the notation of
\cite[(1)]{suppl} we put $ w_1 = \alpha$, $ w_0 = 0.7 \alpha $ and $ \varphi = 0 $.
Remarkably, the low magic $ \alpha$'s of the chiral model seem to provide a good
approximation for the nearly flat bands of the Bistritzer--MacDonald model.
For completeness, the bands for the anti-chiral model $ w_1 = 0 $, $ w_0 = \alpha $, 
$ \varphi = 0 $,  can
be found at \url{https://math.berkeley.edu/~zworski/antichiral_bands.mp4}. As shown 
in \cite{suppl} there are no exact flat bands in that case.}
\end{figure}

\subsection{Symmetries revisited} 

We already recalled that $ \mathscr L_\gamma $ defined in \eqref{eq:FL_ev} 
commutes with $ D ( \alpha ) $ and (extended diagonally) with $ H( \alpha ) $. 
The rotation 
\[   \Omega u ( z ) := u ( \omega  z) , \ \  u \in \mathscr S' ( \mathbb C; \mathbb C^2 ) , \]
satisfies 
\[  \Omega D ( \alpha ) = \omega D ( \alpha ) \Omega , \]
and produces a commuting action on $ H ( \alpha ) $ as follows
\begin{equation}
\label{eq:defC}  \mathscr C H ( \alpha ) = H( \alpha ) \mathscr C , \ \ \ \mathscr C := \begin{pmatrix}
\Omega &  \ 0 \\
 0 & \bar \omega \Omega \end{pmatrix} : L_{\rm{loc}} ^2 ( \mathbb C ; \mathbb C^4 ) \to L^2_{\rm{loc}} ( \mathbb C ; \mathbb C^4 ) . \end{equation}
We then have
\[  \mathscr L_\gamma \Omega = \Omega \mathscr L_{\omega \gamma} , \ \ \
\mathscr L_\gamma \mathscr C = \mathscr C \mathscr L_{\omega \gamma}, 
\ \ \  \mathscr C \mathscr L_\gamma  =  \mathscr L_{\bar \omega \gamma} \mathscr C.
\]
The chiral symmetry is given by 
\begin{equation}
\label{eq:defW}
\begin{gathered}   H (\alpha) = - \mathscr W H (\alpha) \mathscr W, \ \ \ \mathscr W := \begin{pmatrix}
1 & 0 \\
0 & -1 \end{pmatrix} : \mathbb C^n \times \mathbb C^n \to  \mathbb C^n \times \mathbb C^n , 
\\  
\ \ \ \mathscr W  \mathscr C  = 
 \mathscr C  \mathscr W , 
 %\ \ \mathscr Q \mathscr W = \mathscr W \mathscr Q,
 \  
 \ \ \mathscr L_{\gamma } \mathscr W = \mathscr W \mathscr L_{\gamma }.
\end{gathered}
\end{equation}

We follow \cite[\S 2.1]{beta} combine the $\Lambda $ and $ \mathbb Z_3 $ actions into a group of unitary action which
commute with $ H ( \alpha ) $:
\begin{equation}
\label{eq:defG}  
\begin{gathered}  G := \Lambda \rtimes \ZZ_3 , \ \ 
  \ZZ_3 \ni \ell : \gamma \to \bar \omega^\ell  \gamma , \ \ \ ( \gamma , \ell ) \cdot ( \gamma' , \ell' ) = 
( \gamma + \bar \omega^\ell \gamma' , \ell + \ell' ) ,
\\  ( \gamma, \ell ) \cdot u = 
\mathscr L_{\gamma } \mathscr C^\ell u, \ \ \ u \in L^2_{\rm{loc}} ( \mathbb C ; \mathbb C^4 ) . 
\end{gathered}
\end{equation}
By taking a quotient by $ 3 \Lambda $ we obtain a finite group acting 
unitarily on $ L^2 ( \CC/ 3\Lambda ) $ and commuting with $ H ( \alpha ) $:
\begin{equation}
\label{eq:defG3}
G_3 :=  G/3 \Lambda  = \Lambda /3\Lambda \rtimes \ZZ_3 \simeq \ZZ_3^2 \rtimes \ZZ_3. 
\end{equation}
In addition to the spaces $ L^2_k $ defined in \eqref{eq:FL_ev}, we introduce 
\begin{equation}
\label{eq:defLpk} 
\begin{gathered} 
L^2_{ k,p} ( \mathbb C/\Lambda ; \mathbb C^4 ) := 
\{ u \in L^2_{\rm{loc} } ( \mathbb C ; \mathbb C^4 ) : 
\mathscr L_\gamma \mathscr C^\ell  u = e^{ i \langle k , \gamma \rangle } 
\bar  \omega^{\ell p} u \}, \\  H_{k,p}^s := L^2_{k,p} \cap H^s_{\rm{loc}}, \ \ \
k \in (\tfrac13 \Lambda^*)/\Lambda^* \simeq \mathbb Z^2_3, \ \  p \in \mathbb  Z_3 . 
 \end{gathered}
 \end{equation}
This streamlines the notation of \cite{beta} and concentrates on the most
relevant representations of $ G_3 $. We use the same notation for $ \mathbb C^2 $
valued or scalar functions with $ \mathscr C $ replaced by $ \Omega$. 

We have the following orthogonal decompositions 
\begin{equation}
\label{eq:ortho} 
\begin{gathered}
 L^2 ( \mathbb C/ 3 \Lambda ) = \bigoplus_{ k \in \frac13 \Lambda^*/\Lambda^* }
 L^2_k ( \mathbb C/\Lambda ) , \\
 L^2_k ( \mathbb C/\Lambda ) = \bigoplus_{ p \in \mathbb Z_3 } 
 L^2_{k,p}  ( \mathbb C/\Lambda ), \ \ \ 
 k \in \mathcal K/\Lambda^* , 
 \end{gathered}
 \end{equation}
 where
 \begin{equation}
 \label{eq:defK} 
 \  \mathcal K := \{ k \in \mathbb C: \omega k
 \equiv k \!\! \mod \Lambda^* \} =  \{ K, -K, 0 \} + \Lambda^* . 
 \end{equation}

\noindent
{\bf Remark.} Decompositions \eqref{eq:ortho} do not provide a decomposition 
of $ L^2 ( \mathbb C/3 \Lambda ) $ into representations of $ G_3 $ given by 
\eqref{eq:defG3}. In addition to $ L^2_{k,p} $, $ k \in \mathcal K $ and 
$ p \in\mathbb Z^3 $, we also have two irreducible representations of dimension
three -- see \cite[\S 2.2]{beta}. These representations appear 
in $ \ker_{ L^2 ( \mathbb C/ 3 \Lambda, \mathbb C^2 )} D ( \alpha ) $
when $ \alpha \in \mathcal A $ is simple. Since this observation does not play 
a role in our proofs we do not provide details.

W also recall from \cite[\S 1]{beta} the additional 
symmetry 
\begin{equation}
\label{eq:defsE} \mathscr E D ( \alpha ) \mathscr E^* = - D( \alpha )   , \ \ \    \mathscr E v ( z ) := J v ( - z ) , \ \   J : = \begin{pmatrix}  \ \ 0 & 1 \\ -1  &  0 \end{pmatrix} ,
\end{equation} 
noting that it plays a crucial role in \cite{waz}. We have
%\[ \mathscr E :  L^2_{\rho_{k,\ell} } ( \mathbb C/\Gamma, \mathbb C^2 ) 
%\to  L^2_{\rho_{1-k,\ell} } ( \mathbb C/\Gamma, \mathbb C^2 )  , \]
%so that
\begin{equation}
\label{eq:symmE}
\begin{gathered}    L^2_{{K,\ell} } ( \mathbb C/\Lambda, \mathbb C^2 ) \xrightarrow{ \mathscr E  } 
L^2_{{-K,\ell} } ( \mathbb C/\Lambda, \mathbb C^2 ) 
\xrightarrow{ \mathscr E  } 
L^2_{{K,\ell} } ( \mathbb C/\Lambda, \mathbb C^2 ) , \\
L^2_{{0,\ell} } ( \mathbb C/\Lambda, \mathbb C^2 ) \xrightarrow{ \mathscr E  } 
L^2_{{0,\ell} } ( \mathbb C/\Lambda, \mathbb C^2 ) . 
% \mathscr E : L^2_{  \rho_{(1,0)} } ( 
%\CC / \Gamma; \CC^2 )  \to L^2_{  \rho_{(2,0)} } ( 
%\CC / \Gamma; \CC^2 ) .
\end{gathered}
\end{equation}
Finally we recall the anti-linear symmetries
\begin{equation}
\label{eq:H2Q}  
\begin{gathered}
Q v ( z ) = \overline{ v ( - z ) } , \ \ \ 
\mathscr Q u ( z ) := 
\begin{pmatrix} 0 & Q \\ 
Q & 0 \end{pmatrix} u ( z ), \\
Q D( \alpha ) Q =  D (\alpha )^* , \ \ \ 
H(\alpha)\mathscr Q=\mathscr Q H (\alpha),
 \end{gathered}  \end{equation}
 and
 \begin{equation} 
\label{eq:mapQ}
\begin{split}
& Q : L^2_{ k, p } ( \mathbb C/\Lambda ; \mathbb C^2 ) \to 
 L^2_{ k, -p } ( \mathbb C/\Lambda ; \mathbb C^2 ) ,  \\
&  \mathscr Q :
 L^2_{ k, p } ( \mathbb C/\Lambda ; \mathbb C^4 ) \to 
 L^2_{ k, -p+1 } ( \mathbb C/\Lambda ; \mathbb C^4 ) ,   \end{split} \end{equation}
for  $ k \in \mathcal K$, $  p \in \mathbb Z_3 $. 
 For another useful antilinear symmetry see \cite[\S 3.5]{dynip}.

\subsection{Bloch--Floquet theory}
\label{s:BF}

In \cite{beta}, the band theory was based on lattice of periodicity of $ D ( \alpha) $
and $ H ( \alpha ) $ given by $ 3 \Lambda $ (see Appendix A a
translation of notations). That meant that Bloch eigenvalues were functions of $ k \in 
\mathbb C/ \frac 13\Lambda^*$, a small torus. 
Following the physics literature we now consider Bloch--Floquet theory
based on \eqref{eq:FL_ev}, using the commuting operators $ \mathscr L_\gamma $. 
The two approaches are equivalent but Figure \ref{f:bands} 
illustrates the advantages of the latter: the bands have a much cleaner structure
and eigenvalues are functions on a larger torus, $ \mathbb C / \Lambda^* $. 

We first recall that the eigenvalues in \eqref{eq:FL_ev} are the same as 
the eigenvalues of
\begin{equation} \label{eq:defHk} \begin{gathered} 
H_k ( \alpha )  : H^1_0 ( \mathbb C/\Lambda ; \mathbb C^4 ) 
\to  L^2_0 ( \mathbb C/\Lambda ; \mathbb C^4 ) , \\
( H_k ( \alpha )  - E_j ( \alpha, k ) )e_j ( \alpha, k ) = 0 , \ \ 
e_j ( \alpha ,  k ) \in H^1_0  ( \mathbb C/\Lambda ; \mathbb C^4 )  , 
\\
H_k ( \alpha )  := e^{-i \langle z, k \rangle } H ( \alpha ) 
e^{  i \langle z, k \rangle} = \begin{pmatrix} \ \ 0 & D ( \alpha )^* + \bar k \\
D( \alpha ) + k & \ \ 0 \end{pmatrix}
. \end{gathered}
\end{equation}
The eigenvalues of $ H_k ( \alpha ) $ on $ L^2_0 ( \mathbb C /\Lambda; \mathbb C^4 )$
(with the domain given by $ H^1_0 ( \mathbb C/\Lambda; \mathbb C^4 ) $ -- see \eqref{eq:FL_ev}) 
are given by \eqref{eq:eigs}. We note that
\begin{equation}
\label{eq:symeigs} 
\begin{gathered}
E_j ( \alpha, k + p ) = E_j ( \alpha, k ) , \  \ p \in \Lambda^* , \ \ \ 
E_j ( \alpha , \omega k ) = E_j ( \alpha, k ) , \ \ \ k \in \mathbb C . 
\end{gathered}
\end{equation}
The last property follows from checking that
$ \mathscr C H_{ \omega k } ( \alpha ) \mathscr C^* = H_k ( \alpha ) $, 
where $ \mathscr C $ was defined in \eqref{eq:defC}. This shows that $ k \mapsto E_j ( \alpha, k ) $
is either singular or critical at $ K , - K , 0 $ ($ K = 4 \pi /3 $ -- see \eqref{eq:propU} and 
the end of Section \ref{s:ccc}; that is also nicely seen in the animation \url{https://math.berkeley.edu/~zworski/KKmovie.mp4}.).

The key fact used in \cite{magic} and \cite{beta} 
is the existence of protected states. We recall it in the current convention:
\begin{prop}
\label{p:prote}
For every $ \alpha \in \mathbb C $ there exists $ u_{\pm K} ( \alpha ) \in  H^1_{0 } ( \mathbb C /\Lambda; \mathbb C^2 ) $ such that
$ \tau ( K ) u_K ( 0 ) = \mathbf e_1 $, $ \tau(-K ) u_{-K} ( 0 ) = \mathbf e_2 $, 
\begin{equation} \label{eq:deftau} 
\begin{gathered}
 ( D( \alpha ) \pm K ) u_{ \pm K} ( \alpha ) = 0 , \ \ \ 
\tau ( k ) v ( z ) := e^{ i \langle z, k \rangle } v ( z ) , 
\end{gathered} 
\end{equation}
%\blue{ Or rather (???)
%\begin{equation} \label{eq:deftau} 
%\begin{gathered}
% ( D( \alpha ) \pm K ) u_{ \pm K} ( \alpha ) = 0 , \ \ \ 
%\tau ( k ) v ( z ) := e^{ i \langle z, k \rangle } v ( z ) , 
%\end{gathered} 
%\end{equation}}
where we note that $ \tau ( k ) : L^2_{p} \to L^2_{ p+k} $, $ p , k \in \mathbb C $. In addition, 
%\blue{Something is strange here: we do want $ D ( \alpha ) (\tau ( \pm K ) u_{\pm K } ) = 0$ 
%but that suggests a different sign in the definition!} 
\begin{equation}
\label{eq:protau}   \tau( \pm K ) u_{\pm K } ( \alpha ) \in { 
\ker_{ H^1_{ \pm K,  0 } ( \mathbb C/\Lambda; \mathbb C^2 ) } D ( \alpha )} , \end{equation} 
and if 
$ \tau(\pm K)   u_{\pm K } ( \alpha, z ) = ( u_1^\pm ( \alpha, z ) , u_2^\pm ( \alpha , z ) )^t $ then 
\begin{equation}
\label{eq:defzS}  u_2^+ ( \alpha,  \pm z_S ) = u_1^- ( \alpha, \pm z_S ) = 0 , \ \ 
z_S := i / \sqrt 3, \ \ \ \omega z_S = z_S - ( 1 + \omega ) . \end{equation}
%(This follows from the proof of existence of protected states -- see \cite[\S 2.2]{beta} or 
%\cite[Proposition 3.2]{dynip}; we have $ \tau ( K ) u_K ( 0 ) = ( 1, 0 )^t \in L^2_{K,0} $
%and $ \tau (-K ) u_{-K} ( 0 ) = ( 0,1 )^t \in L^2_{-K, 0 } $.)
\end{prop}
\begin{proof}
We  decompose $ \ker_{H^1 ( \mathbb C / 3 \Lambda; \mathbb C^{4} )}
H ( \alpha) $ into representation of $ G_3 $ (see \eqref{eq:defG3}
and \cite[\S 2.2]{beta} for a review
of representations of $ G_3 $ -- we only use representation appearing in \eqref{eq:ortho}
so that is all that is needed here).
From \eqref{eq:defW} we see that the spectrum of $ H( \alpha ) $ restricted to representations of
$ G_3 $ is symmetric with respect to the origin. 
The kernel of $ H ( 0 ) $ on $ H^1 ( \mathbb C / 3 \Lambda; \mathbb C^{4} )$ is given 
by by the standard basis vectors in $ \mathbb C^4 $, $ \mathbf e_j $. They satisfy
\[  \mathbf e_1 \in H^1_{ {K, 0} }, \ \  \mathbf e_2 \in H^1_{{-K, 0 } }, \ \ 
\mathbf e_3 \in H^1_{{K, 1} } , \ \  \mathbf e_4 \in H^1_{{\blue{ - } K , 1 } }, \]
and all these spaces are mutually orthogonal. 
Since the spectrum of $ H ( \alpha ) |_{  L^2_{ {k,p} } } $ is even, 
continuity of eigenvalues shows that
$   \dim \ker_{ L^2_{\pm K, p } } H( \alpha ) \geq 1$, $ \alpha \in \mathbb C $, 
$ p = 0,1 $. 
Since $  \tau(\mp K) : \ker_{ H^1_{\pm K } } H( \alpha ) \to \ker_{ H^1_0 } ( H ( \alpha ) \pm  K ) $ 
this gives \eqref{eq:deftau} and \eqref{eq:protau}. 

For \eqref{eq:defzS} we give an argument in the case of $ u_K $: 
$ u_2 ( \pm z_S ) = u_2 ( \pm \omega z_S ) = u_2 ( \pm z_S \mp ( 1 + \omega ) ) $
and in view of $ ( u_1, u_2 )^t \in L^2_{ K } $, the right hand side is
equal to $  e^{ \mp 2 i \langle ( 1 + \omega ), K \rangle }  u_2 ( \pm z_S ) $. Since 
$ e^{ i \langle 1 + \omega, K \rangle } = e^{  \frac 43 i \pi \Re ( 1 + \omega ) } = 
\omega $, we see that $ u_2 ( \pm z_S ) = 0 $.
\end{proof}

As a consequence of Proposition \ref{p:prote} we have
\begin{equation}
\label{eq:defK0}
\forall \, \alpha \in \mathbb C  \ \ \Spec_{ L^2_0 ( \mathbb C /\Lambda; \mathbb C^2 )} D( \alpha ) \supset \mathcal K_0 := 
\{ K, -K \} + \Lambda^* . 
\end{equation}

\subsection{Spectral characterization of magic angles}
\label{s:specc}

In \cite{magic} magic angles were
computed by analyzing $ u_{\pm K } $ (see Proposition \ref{p:prote})
 and identifying $ \mathcal A $ with the 
zeros of the Wronskian, 
\begin{equation}
\label{eq:Wron}  v ( \alpha ) =  W ( \tau(K ) u_K ( \alpha ) , \tau(-K) u_{-K } ( \alpha ) )  , \ \ \
W ( v , w ) := \det ( u , v ) , \ \ u , v \in \mathbb C^2.  \end{equation}
The function $ \alpha \mapsto v ( \alpha ) \in \mathbb C  $ was also identified with the 
physical quantity called the {\em Fermi
velocity} \cite[(7),(8)]{magic}. That led 
to a rough 
(three digits) computation of the first five $ \alpha$'s \cite{magic}
and then 
a computer assisted rigorous proof of the
existence of the first magic $ \alpha $ \cite{lawa}. Proposition \ref{p:protea} below
shows that we can choose $ u_K ( \alpha ) $ so that $ v ( \alpha ) $ in an entire function.

The approach taken in \cite{suppl,beta} was different and was based on 
identifying magic $ \alpha $'s with reciprocals of eigenvalues of a family of
compact operators. Crucially, the eigenvalues are independent of the 
elements of the family and that lies behind Theorem \ref{t:1}. 
 We recall this
in a form generalizing \eqref{eq:defD}:
\begin{equation}
\label{eq:newU}
\begin{gathered} 
D_{V } ( \alpha ) := 2 D_{\bar z } + {\alpha} V ( z ) , \ \ \ 
V ( z ) = \begin{pmatrix} 0 & U_+ ( z ) \\
U_- ( z ) & 0 \end{pmatrix} , \ \ \  U_\pm ( \omega z ) = \omega U_\pm ( z ) , 
 \\
U_\pm ( z + \gamma ) = e^{ \pm i \langle \gamma, K \rangle}  U_\pm ( z ) , \ \ 
\gamma \in \Lambda . 
\end{gathered}
\end{equation}
We also define $ H ( \alpha ) $ and note that the results of the previous
sections apply without modification. We define the set $ \mathcal A := 
\mathcal A ( V ) $ by \eqref{eq:defA}. Since
\[   \begin{pmatrix} 1 & \ \ 0 \\ 0 & -1 \end{pmatrix} D_{V}
 ( \alpha ) \begin{pmatrix} 1 & \ \ 0 \\ 0 & -1 \end{pmatrix} = D_{V}
 ( - \alpha ) , \]
we still have the symmetry $ \mathcal A ( {V}
 ) = - \mathcal A ( {V}
 ) $. 

If in addition to \eqref{eq:newU} we also have 
\begin{equation}
\label{eq:Vreality}  \overline{ U_\pm (  \bar z ) } = - U_\pm ( {-}z )  \ \Longleftrightarrow \ 
V ( z ) = -  \overline{V({-}\bar z)} , 
\end{equation}
then we get
$  \widetilde \Gamma D_{V}
 ( \alpha ) \widetilde \Gamma = - D_{V}
 ( - \bar \alpha )$, 
$  \widetilde \Gamma {v} ( z ) := \overline{ {v} ( \bar z )} $, 
 and hence $ \mathcal A ( {V}
 ) = \overline{ \mathcal A ( {V}
) } $. 

The following result is a generalized formulation of \cite[Theorem 2]{beta}. To state it 
we define 
\begin{equation}
\label{eq:defres} 
\begin{gathered}  R ( k ) := ( 2 D_{\bar z } - k )^{-1} : L^2_p ( \mathbb C/\Lambda, \mathbb C^2 ) :
\longrightarrow L^2_p ( \mathbb C/\Lambda ; \mathbb C^2 ) ,  \ \ \ p \in \mathbb C , \\
k \notin \mathcal K_0 + p  , \ \ \  \mathcal K_0 := \{ K, - K \} + \Lambda^* . 
\end{gathered} 
\end{equation}
This follows from the fact that $ ( 2 D_{\bar z } - k )^{-1} : L^2_0 \to L^2_0 $ for $
k \notin \mathcal K_0 $. We then 
have $ \tau ( p ) : L^2_0 \to L^2_p $ and $ \tau(- p ) R ( k ) \tau  ( p ) =
R ( k - p ) $. 

\begin{prop}
\label{p:spec}
In the notation of \eqref{eq:newU} and \eqref{eq:defres} the following compact 
operators are well defined
\begin{equation}
\label{eq:defTk}  T_k := R ( k ) V :  L^2_p ( \mathbb C/\Lambda, \mathbb C^2 ) \to 
L^2_p ( \mathbb C/\Lambda, \mathbb C^2 ) , \ \ k \notin \mathcal K_0 + p . 
\end{equation}
Moreover, 
\begin{equation}
\label{eq:p2q}  
\Spec_{ L^2_p } ( T_k ) = \Spec_{ L^2_q } ( T_{k' } ) , \ \ 
k \notin \mathcal K_0 + p , \ \ k' \notin \mathcal K_0 + q , \end{equation}
\begin{equation}
\label{eq:spect}
\mathcal A ( V ) = \{ \alpha \in \mathbb C :  \alpha^{-1} \in \Spec_{ L^2_p }( T_k  ) \}, \ \ 
k \notin p + \mathcal K_0 , \ \ p \in \mathbb C,
\end{equation}
and
\begin{equation}
\label{eq:A2S}
\Spec D_V ( \alpha ) = \left\{ \begin{array}{ll} \mathcal K_0 , & \alpha \notin \mathcal A ( V ) \\
\mathbb C, & \alpha \in \mathcal A ( V ) , \end{array} \right. 
\end{equation}
with simple eigenvalues when $ \alpha \notin \mathcal A( V ) $. 
\end{prop} 
\begin{proof}
We first note that the definition of $ \mathscr L_\gamma $ 
in \eqref{eq:FL_ev} and \eqref{eq:newU} show that $ \mathscr L_\gamma V = 
 V \mathscr L_\gamma $ and hence $ V : L^2_p \to L^2_p $. This and \eqref{eq:defres}
give the mapping property \eqref{eq:defTk}.
% \[ \begin{split} \mathscr L_\gamma (V u ) ( z ) & = \begin{pmatrix} 
%e^{  i \langle \gamma , K \rangle } U_+ ( z + \gamma ) u_2 ( z + \gamma ) \\
%e^{   - i \langle  \gamma , K \rangle } U_- ( z + \gamma ) u_1 ( z + \gamma ) \end{pmatrix}
%= \begin{pmatrix} e^{ 2 i \langle \gamma , K \rangle } U_+ ( z ) u_2 ( z + \gamma )\\
%e^{ - 2 i \langle  \gamma, K \rangle } U_- ( z ) u_1 ( z + \gamma ) \end{pmatrix} \\
%& = V \begin{pmatrix} e^{ i \langle \gamma, K \rangle} u_1 ( z+ \gamma ) \\
%e^{ - i \langle \gamma, K \rangle} u_2 ( z+ \gamma ) \end{pmatrix} = ( V \mathscr L_\gamma )
%( z ) . \end{split} \]

We first consider  \eqref{eq:p2q}  and \eqref{eq:spect} for $ q= p = 0 $ and $ k \notin \mathcal K_0 $ 
(this spectral characterization was proved in \cite{beta} but we include a streamlined proof using the 
current convention). For a fixed $ k \notin \mathcal K_0 $, 
we define a discrete set $ \mathcal A_k := \{ \alpha \in \mathbb C  : - \alpha^{-1} \in \Spec_{L^2_0 }  T_k \} $. 
For $ \alpha \notin \mathcal A_k $ the spectrum of $ D ( \alpha ) $ is then discrete since
\begin{equation}
\label{eq:defKk} D ( \alpha ) - z = ( D (0 ) - k ) ( I + K ( z ) ) , \ \ \ K ( z) :=  \alpha T_k + R ( k ) ( k - z )  , 
\end{equation}
and $ z \mapsto   K ( z) $ is a holomorphic family of compact 
operators with $ I + K ( k ) $ invertible (since $ - \alpha^{-1} \notin \Spec_{L^2_0 } ( T_k ) $).
But that implies (see for instance \cite[Theorem C.8]{res}) that $ (D ( \alpha ) - z )^{-1} = ( I+ K ( z) )^{-1} R ( k ) $ is a meromorphic family of operators, and in particular, the spectrum of $ D ( \alpha ) $ is
discrete. 

We now put 
\[  \Omega := \{ \alpha \in \mathbb C \setminus \mathcal A_k : \Spec_{L^2_0} ( D ( \alpha ) ) = 
\mathcal K_0 \text{ with simple eigenvalues} \} , \]
noting that $ 0 \in \Omega $. We  claim that $ \Omega $ is open and closed in the
relative topology of the connected topological space $ \mathbb C \setminus \mathcal A_k $.
That will imply that $ \Omega = \mathbb C \setminus \mathcal A_k $. 

To prove the claim, we note that for % we recall from \eqref{eq:persp} that the spectrum is periodic,
%$ \Spec_{L^2_0} ( D ( \alpha ) ) =  \Spec_{L^2_0} ( D ( \alpha ) ) + \Lambda^* $. If
$ \alpha_0 \in \Omega $ there exists a neighbourhood of $ \alpha_0 $, $U $, 
such that for $ \alpha \in U $, the spectrum of $ D ( \alpha ) $ is discrete and changes
continuously with $ \alpha $. From Proposition \ref{p:prote} we also know that
$ \Spec_{L^2_0 } ( D ( \alpha ) ) \supset \mathcal K_0 $. But as it is equal to $ \mathcal K_0 $
at $ \alpha = \alpha_0 $ it has be equal to $ \mathcal K_0 $ in $ U $. To see that $ \Omega $
is closed, assume that $ \{ \alpha_j \}_{j=1}^\infty 
\in \Omega $, $ \alpha_j \to \alpha_0 \in \mathbb C \setminus 
\mathcal A_k $. But this means that there exists an open neighbourhood of $ \alpha_0 $,  $ U $, 
such that for $ \alpha \in U $ the spectrum is discrete and hence depends continuously
on $ \alpha $. Since $ \Spec_{L^2_0 } ( D ( \alpha_j ) ) = \mathcal K_0 $, 
we conclude that $ \Spec_{L^2_0 } D ( \alpha_0 ) = \mathcal K_0  $ (all with agreement of 
simple multiplicities),
that is, $ \alpha_0 \in \Omega $. 

It remains to show that $ \mathcal A_k $ is independent of $ k $. For that we note that
 $ -\alpha^{-1}  \in \Spec_{L^2_0 } T_k $ is equivalent to $ k \in \Spec_{L^2_0 } ( D ( \alpha ) ) $
(see $ K ( k ) $ in \eqref{eq:defKk}). Since $ k \notin \mathcal K_0 $ the spectrum cannot
be discrete, as then it would be equal to $ \mathcal K_0 $. Hence, it has to be equal to 
$ \mathbb  C $ (if there were any points at which $ D ( \alpha ) - z $ were invertible then 
the compactness of the inverse and an argument similar to that after \eqref{eq:defKk} 
would show the spectrum is discrete). But that means that any $ k' \in \Spec_{L^2_0 } ( D ( \alpha ) )$ 
and the equivalence above shows that $ -\alpha^{-1} \in \Spec_{L^2_0} ( T_{k'} ) $.

 In particular, 
$ \Spec_{ L^2_0 } ( T_{k_1} ) = \Spec_{ L^2_0 } ( T_{k_2} )  $ for any $ k_j \notin \mathcal K_0 $.
To establish \eqref{eq:p2q} we can take $ q = 0 $ and note that that $ \tau ( -p ) : L^2_p \to L^2_0 $ (see \eqref{eq:deftau}) $ \tau(p) T_{k_1} \tau( p)^{-1} = T_{k_1 + p } : L_p^2 \to {L^2_p} $, 
$ k_1 \notin \mathcal K_0 $ (and hence $ k_1 + p \notin \mathcal K_0+ p $). Hence
to see \eqref{eq:p2q} with $ q = 0 $ we take $ k_2 = k' $ and $ k_1 = k - p $. 
\end{proof}

\begin{proof}[Proof of Theorem \ref{t:1}]
This is immediate from from \eqref{eq:A2S}: if $ E_1 ( \alpha, k ) = 0 $ 
for $ k \notin \mathcal K_0 $, $ 0 \in \Spec_{ L^2_0 } H_k ( \alpha ) $,  
then  $ \ker_{ H^1_0}  (  D ( \alpha ) + k ) $ or 
$ \ker_{ H^1_0 } ( (D ( \alpha )+ k)^*) $ are non zero. Since, $ D ( \alpha ) + k $ is
a Fredholm operator of index zero (see \cite[Proposition 2.3]{beta}) the two 
statements are equivalent. But $ \ker ( D ( \alpha ) + k) \neq \{ 0 \} $, $ k \notin \mathcal K_0 $,
implies in view of \eqref{eq:A2S} that $ E_1 ( \alpha, k ) \equiv 0 $, $ k \in \mathbb C $.
\end{proof}

Combining Propositions \ref{p:prote} and \ref{p:spec} we obtain a stronger statement
about protected states:
\begin{prop}
\label{p:protea}
Suppose that $ D ( \alpha ) $ is given by \eqref{eq:defD}, with $ U $ satisfying
\eqref{eq:propU}. Then,
for $ \alpha \notin \mathcal A $, $ u_{\pm K } ( \alpha ) $ are unique up to 
multiplicative constants, and we can choose
\begin{equation}
\label{eq:+2-}
u_{-K} ( \alpha ) = \tau ( K ) \mathscr E \tau ( K ) u_{ K } ( \alpha) .
\end{equation} 
Moreover, $ \alpha \mapsto u_{\pm K } ( \alpha ) $ can be chosen to be holomorphic
as a function of $ \alpha \in \mathbb C $ with values in $ H^1_0 ( \mathbb C/\Lambda; 
\mathbb C^2 ) $.
\end{prop}
\begin{proof} 
Since for $ \alpha \notin \mathcal A $, the eigenvalues of $ D( \alpha ) $ are simple
and the right hand side of \eqref{eq:+2-} has all the properties of $ u_{-K} ( \alpha )  $ in 
Proposition \ref{p:prote}, we can choose it to be $ u_{-K } ( \alpha ) $.

To find a holomorphic family $ \alpha \mapsto u_K ( \alpha )$ we proceed as follows. 
We first note that for $ \alpha_0  \notin \mathcal A $,  $( \tau ( K ) u_K ( \alpha_0 + \zeta ) , 0 )^t $
spans $ \ker  \widetilde H( \alpha_0 , \zeta  )|_{ H^1_{ K, 0 } } $, 
\[ \widetilde H ( \alpha_0 , \zeta ) := \begin{pmatrix} 0 & D ( \alpha_0 + \bar \zeta )^* \\
D ( \alpha_0 + \zeta ) & 0 \end{pmatrix} . \]
Since $ \zeta \mapsto H ( \alpha_0 , \zeta ) $ is a holomorphic family of operators
it follows that we can choose $ \zeta \mapsto  \tau ( K ) u_K ( \alpha_0 + \zeta ) $ holomorphic in 
$ \zeta $ for $ |\zeta | < \delta $ (note that $ ( \tau ( K ) u_K ( \alpha_0 + \zeta ) , 0 )^t 
\in \ker H ( \alpha_0 + \zeta ) |_{ H^1_{ K, 0 } } $). When $ \alpha_0 \in \mathcal A $, 
 $ \zeta \mapsto H ( \alpha_0, \zeta )$ is a holomorphic family of operators, 
which is self-adjoint for $ \zeta \in \mathbb R $. Rellich's theorem 
\cite[Chapter VII, Theorem 3.9]{kato}, then shows that an element of the kernel of
$ H( \alpha_0 , \zeta  )|_{ H^1_{ K, 0 } } $ can be chosen to be holomorphic near $ 
\zeta = 0 $. In view of simplicity for $ 0 < |\zeta| < \delta $, it has to 
coincide with a choice of $ \tau( K ) u_K ( \alpha_0 + \zeta) $. 

The local constructions above and a partition of unity on $ \mathbb C $ show
that we can choose $ \tau ( K ) \widetilde u_K \in C^\infty ( \mathbb C ; H^1_{K,0} )$
and it remains to modify it so that it becomes holomorphic. We have
\[  0 =  \partial_{\bar \alpha } ( D ( \alpha ) \tau( K ) u_K ( \alpha ) ) = 
D ( \alpha ) ( \tau ( K ) \partial_{\bar \alpha } u_K ( \alpha ) ) . \]
For $ \alpha \notin \mathcal A $ the kernel on $ H^1_K $ is one dimensional 
and hence 
\begin{equation}
\label{eq:falfa} \partial_{\bar \alpha } \widetilde u_K ( \alpha ) =  f( \alpha ) \widetilde 
u_K ( \alpha ) , \ \ \alpha \notin \mathcal A , \ \ \ f ( \alpha ) = 
\frac{ \langle  \partial_{\bar \alpha }\widetilde u_K ( \alpha ), 
  \widetilde u_K ( \alpha )\rangle}{ \| u_K ( \alpha ) \|^2}  . \end{equation}
In the formula for $ f ( \alpha ) $, the right hand side is smooth in $ \alpha $ and that 
shows that the first formula in \eqref{eq:falfa} holds for all $ \alpha \in \mathbb C $. 
The equation $ \partial_{\bar \alpha} F ( \alpha )  = f ( \alpha ) $ 
(see for instance \cite[Theorem 4.4.6]{H1} applied with $ P = \partial_{\bar \alpha } $
and $ X = \mathbb C $)
can be solved 
with $ F \in C^\infty ( \mathbb C ) $. This shows that 
$  u_K ( \alpha ) = \exp ( - F ( \alpha ) ) \widetilde u_K ( \alpha ) $ is indeed 
holomorphic.
\end{proof}

\section{Theta function argument revisited}
\label{s:theta}

In \cite{magic} a theta function argument was used to explain the formation of
flat bands and in \cite{beta} that approach was shown to be equivalent to the 
spectral characterisation  in Proposition \ref{p:spec}. We review it here from the point
of view of \S \ref{s:specH} and \cite{led}, where the holomorphic dependence of
eigenvectors on the Floquet parameter (Bloch pseudo-momentum) $  k $ was stressed.

\subsection{Theta functions}
\label{s:theta} 
To simplify notation we put
$ \theta ( z ) :=  \theta_{1} ( z | \omega ) := - \theta_{\frac12,\frac12} ( z | \omega ) $, 
and recall that 
\begin{equation}
\label{eq:theta}
\begin{gathered} 
\theta ( z )
= - \sum_{ n \in \mathbb Z } \exp ( \pi i (n+\tfrac12) ^2 \omega+ 2 \pi i ( n + \tfrac12 ) (z + \tfrac
12 )  ) , \ \ \ \theta ( - z ) = - \theta ( z ) \\
\theta  ( z + m ) = (-1)^m \theta   ( z ) , \ \ \theta  ( z + n \omega ) = 
(-1)^n e^{ - \pi i n^2 \omega - 2 \pi i z  n } \theta   ( z ) ,
\end{gathered}
\end{equation}
and that $ \theta $ vanishing simply on $ \Lambda $  and nowhere else 
see \cite{tata}. 

We now define 
\begin{equation}
\label{eq:defFk}  F_k ( z ) = e^{\frac i 2   (  z -  \bar z ) k } \frac{ \theta ( z - z ( k ) ) }{
\theta( z) } ,  \ \ \ z (k):=  \frac{ \sqrt 3 k }{ 4 \pi i } , \ \  z:  \Lambda^* \to \Lambda . 
\end{equation}
Then, using \eqref{eq:theta} and differentiating in the sense of distributions, 
\begin{equation}
\label{eq:propFk}
\begin{gathered}F_k ( z + m + n \omega ) = e^{  -  n k \Im \omega } 
e^{ 2 \pi i n z ( k ) } F_k ( z ) = F_k ( z ) , \\
  ( 2 D_{\bar z } + k ) F_k ( z ) %= e^{ \frac12  i (  z -  \bar z ) k } 
%2 D_{\bar z} \left( \frac{ \theta ( z - z ( k ) ) }{
%\theta( z) }\right)
 = c(k)   \delta_0 ( z ) , \ \ c(k) :=  2 \pi i {\theta ( z ( k ) ) }/{ \theta' ( 0 ) } .
\end{gathered} \end{equation}
This follows from the fact that $ 1/(\pi z ) $ is a fundamental solution 
of $ \partial_{\bar z } $ --  see for instance \cite[(3.1.12)]{H1}.  In other, 
words, for $ k \notin \Lambda^* $, $ F_k $ gives the Green kernel of 
$2 D_{\bar z }  + k $ on the torus $ \mathbb C/\Lambda $:
\[   ( D_{\bar z } + k )^{-1} f  ( z ) = c(k)^{-1} \int_{ \mathbb C/\Lambda } 
F_k ( z - z' ) f ( z' ) dm ( z ) ,\]
$  d m ( z ) = dx dy $,  $  z= x+ iy$. 
For future use we record some properties of $ F_k $:
\begin{lemm}
\label{l:van1}
\label{eq:FkFp}
 For $  u \in C^\infty ( \mathbb C )$, we have, in the sense of distributions, 
 and in the notation of \eqref{eq:propFk}, 
\begin{equation}
\label{eq:delta1} 
\begin{split}  &  ( 2  D_{\bar z } + q - \ell ) \left(\frac{ F_q ( z - z_0 )}{F_\ell ( z - z_0 ) } u ( z ) \right) = 
\frac{ F_q ( z - z_0 )}{F_\ell ( z - z_0 ) } 2 D_{\bar z } u ( z )  + c(q, \ell ) u ( z_1 )  \delta (z - z_1 )   , \end{split}  \end{equation} 
where $ z_1 = z(p) + z_0 
 $ and $c ( k , p ) = 2 \pi i \theta ( z ( q - \ell ) ) /  \theta' ( 0 ) $.
In particular, by taking $ \ell = 0 $ and $ q = k $, 
\begin{equation}
\label{eq:delta}   ( 2  D_{\bar z } +  k ) \left(F_k ( z - z_0 ) u ( z ) \right) = F_k ( z - z_0 ) 2 D_{\bar z } u ( z ) + c(k ) u ( z_0 ) 
\delta ( z- z_0 )  .  \end{equation} 
\end{lemm}

The following simple lemma is implicit in \cite{magic}:
  \begin{lemm}
 \label{l:van}
 Suppose that that $   w \in C^\infty ( \mathbb C; \mathbb C^2 ) $
 and  that $ ( D ( \alpha ) +   k )   w = 0 $ for some $   k $ and 
 that $   w ( z_0 ) = 0 $. Then  
 $     w ( z ) = ( z - z_0 )   w_0 ( z ) $, where $ 
   w_0 \in C^\infty ( \mathbb C; \mathbb C^2 ) $. 
 \end{lemm}
 \begin{proof}
 The conclusion of the lemma is equivalent to $ (2 D_{\bar z})^\ell   w ( z_0) = 0$ for 
 all $ \ell $. Since 
 $  (2 D_{\bar z })^\ell     w ( z )  =  (2 D_{\bar z })^{\ell-1}  \left[ \left(    U -   k 
 \right)    w \right]  ( z )  $
 that follows by induction on $ \ell $. 
 \end{proof}

 These two lemmas are the basis of the {\em theta function argument} in \cite{magic} (see also
\cite{dun} for an earlier version of a similar method). Suppose $ D ( \alpha ) u = 0 $, $ u 
\in H^1_0 $. 
and $ u ( z_0 ) = 0 $. Lemma \ref{l:van} shows that, near $ z_0 $, $ u ( z ) = ( z - z_0 ) w ( z ) $, 
$ w \in C^\infty $. But then \eqref{eq:delta} shows that
\[  ( D ( \alpha ) + k ) ( F_k ( z -z_0 ) u ( z ) ) = 0 , \ \  F_k ( z -z_0 ) u ( z ) \in H^1_0 , \]
and from an element of the kernel of $ u $ on $ H_0^1 $ we obtained eigenfunction in 
$ H^1_0 $ for all $ k $. (Strictly speaking we do not even need Lemma \ref{l:van} since
elliptic regularity guarantees smoothness of $ z \mapsto F_k ( z -z_0 ) u ( z ) $.)

%\red{Written like this, this looks like a technical lemma but I think this a very key lemma in the paper. Should we add a remark to explain that the true meaning of formula \eqref{eq:delta} is to construct a new eigenfunction  from a previous one by using where it vanishes. I believe this is what is refered as the "theta function argument" in some part of the paper and it might be good to explain it (at least heuristically). This would actually explains why a paper on the simplicity of magic angles focuses so much on the zeros of eigenfunctions.}

We will also need the properties of $ F_k $ when $ k $ is translated, this will allow us to define a natural hermitian line bundle over $\mathbb C/\Lambda$ studied in subsection \ref{s:ccc}:
\begin{lemm}
\label{l:Fk2p}
For $ p \in \Lambda^* $, 
\begin{equation}
\label{eq:Fk2p}
\begin{gathered}
F_{ k + p } ( z ) = e_{p} ( k)^{-1} \tau ( p )^{ {-1}} F_k ( z ) , \\ e_p ( k ) := \frac{\theta ( z ( k ) ) }{ \theta (z ( k + p )) } = ( -1 )^n ( -1)^m e^{ i \pi n^2 \omega + 2 \pi z ( k ) }, 
\end{gathered} 
\end{equation}
where $  z ( p ) = m + n \omega$, $n, m \in \mathbb Z$. 
\end{lemm}
\begin{proof}
Since, for $ k \notin \Lambda^* $,  $ ( 2D_{\bar z } + k + p ) \tau ( p )^{ {-1}} c(k)^{-1} F_k = \delta_0 $
and $ ( 2 D_{\bar z } + k + p ) c( k + p )^{-1} F_{k+p} = \delta_0 $,
the uniqueness of the kernel of the resolvent of $ 2 D_{\bar z }$ shows that
\[ F_{ k+p } ( z) =  \frac{ c ( k + p)}{ c( k ) } [ \tau( p)^{ {-1}}  F_k ] ( z )  = \frac{\theta ( z ( k + p ) ) }{\theta ( z ( k ) )}
[ \tau( p)^{ {-1}} F_k ] ( z ) , \]
and \eqref{eq:Fk2p} follows.
\end{proof}

\subsection{Flat bands and theta functions}
\label{s:flat}

We now reformulate the characterization of magic angles using the vanishing of 
$ u_K ( \alpha ) $ (in \cite{beta} this was established only for $ \alpha \in \mathbb R $): 
\begin{prop}
\label{p:vanish}
Let $ \alpha \mapsto u_K ( \alpha ) \in 
\ker_{ H_0^1 ( \mathbb C/\Lambda , \mathbb C^2 ) } ( D ( \alpha ) + K ) $
%\blue{ Should it be $ \ker_{ H_0^1 ( \mathbb C/\Lambda , \mathbb C^2 ) } ( D ( \alpha ) + K ) $?}
be a smooth family given in Proposition \ref{p:protea}. Then 
\begin{equation}
\label{eq:ukzS} 
\begin{split} \alpha \in \mathcal A & \ \Longleftrightarrow \ \exists \, \varepsilon \in \{ \pm 1 \} \ \ 
u_K ( \alpha , \varepsilon z_S ) = 0 , \ \ \  z_S := i / \sqrt 3 \\
& \ \Longleftrightarrow \ \exists \, z_0 \ \ u_K ( \alpha ,  z_0 ) = 0 .
\end{split} 
\end{equation} 
\end{prop}
\begin{proof}
Suppose first that there exists $ z_0 $ at which $ u_K ( \alpha )  $ vanishes. Since
$ (D ( \alpha ) + K )  u_K = 0 $ we then see that for every $ {k'} \in \mathbb C $
\begin{equation}
\label{eq:K2k}  ( D ( \alpha ) + K + k' ) ( F_{k'} ( z - z_0 )  u_K ( z ) ) = 0 , \end{equation} 
and the solution of this elliptic equation is automatically in $ H^1_0 $ (since $ u_K \in H^1_0 $ and the scalar valued function $ F_{k'} $ is periodic {by \eqref{eq:propFk}}).

 Hence $ \Spec_{ L^2_0 }  D( \alpha ) = 
\mathbb C $. Using \eqref{eq:+2-} and putting
$ u_K = ( u_1, u_2 )^t $,  the Wronskian \eqref{eq:Wron}, which is constant (apply $ \partial_{\bar z } $ 
to both sides and use periodicity), is given by
\[    v ( \alpha ) = u_1 ( z ) u_1 (- z) + u_2 ( -z ) u_2 ( z ) = \left\{ \begin{array}{ll} 
u_1 ( z_0) u_1 ( -z_0 ) + u_2 ( -z_0) u_2 ( z_0 ) = 0 ,\\
\ \  u_1 ( z_S ) u_1 ( -z_S ) , \end{array} \right. \]
where we used \eqref{eq:defzS}. Hence $ u_K $ has to vanish at either $ z_S $ or $ - z_S $.

It remains to show that if $ u_K ( z_S ) u_K ( -z_S ) \neq 0 $ then $ \alpha \notin \mathcal A $.
That is equivalent to the Wronskian, $ v ( \alpha ) \neq 0 $ in which case we can 
express $ ( D ( \alpha ) - k )^{-1} $, $ k \notin \mathcal K$ using $ u_K $ and $ u_{-K} $ -- 
\cite[Proposition 3.3]{beta}.
\end{proof}

\section{Proofs of Theorems \ref{t:HtL} and \ref{t:zS}}
\label{s:prooft}

The theta function argument of \cite{magic} which we reviewed in the previous
section relies on vanishing of both component of
$   u_K  \in L^2_{ 0 } ( \mathbb
C /\Lambda; \mathbb C^2 ) $, $ (  D ( \alpha ) + K )  u_K = 0 $, 
or equivalently of vanishing of $  u_{- K } =\tau ( K )  \mathscr E  \tau ( K) u_{ K}$
 -- see Proposition \ref{p:vanish}, Figure \ref{f:1D} and the movie referenced there. 

We start with a general fact:
\begin{lemm}
\label{l:zeros}
Suppose that $ u \in \ker_{L^2_p } D ( \alpha ) \setminus \{ 0 \} $, $ p \in \mathbb C $, 
has $ k $ zeros (counted with multiplicity). Then 
\[ \dim_{ L^2_p } \ker D ( \alpha ) \geq k . \]
\end{lemm}
\begin{proof}
Define a holomorphic function $ F ( z ) := \prod_{ j=1}^p \theta ( z - z_j ) $. From \eqref{eq:theta} we
see that 
\[  F ( z + \gamma ) = e_\gamma ( z ) F ( z ) , \ \  e_{1 } ( z ) = (-1)^{p},  \ \ 
e_{\omega } ( z ) = e^{ i \beta - 2 \pi i z p },  \]
where $ \beta := - \pi i ( n^2 \omega + 1 ) p - 2 \pi i p \sum_{ j=1}^p 
z_j  $, and where $ z \mapsto e_\gamma( z )  $ satisfies \eqref{eq:CD}. For this $ e_\gamma $ define
\[  \mathscr G := \{  G \in \mathscr O ( \mathbb C ) : G ( z + \gamma ) = 
e_\gamma( z ) G ( z ) , \ \ \gamma \in\Lambda \} \]
which can be interpreted as the space of holomorphic sections for the line bundle defined using
the multiplier $ e_\gamma $ (see \eqref{eq:secti}). The dimension of the vector
space $ \mathscr G $ is given by $ p $ - see  \cite[Proposition 7.9]{notes} for an elementary 
argument (this can be seen from the Riemann--Roch theorem). 

Lemma \ref{l:van} shows that for $ F $ above and any $ G \in \mathscr G $, 
\[  \widetilde u ( z ) := \frac{ G ( z) u ( z ) } { F ( z ) }  \in 
\ker_{ L^2_p } D ( \alpha ) . \]
Hence the dimension of that kernel is at least $ p $.
\end{proof}

\begin{center}
\begin{figure}
\includegraphics[width=15cm]{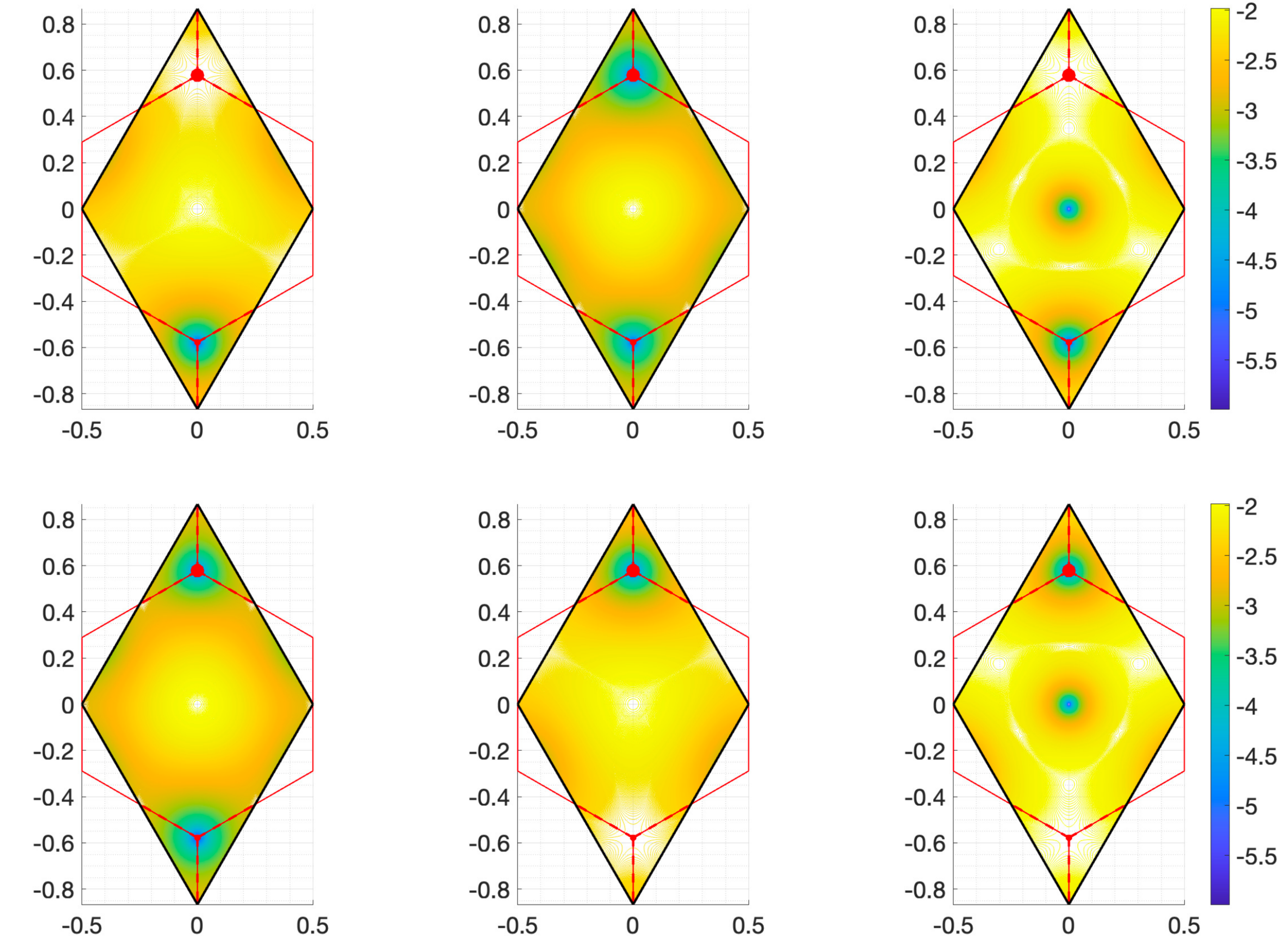}
\caption{\label{f:1D} On top/bottom, the first/second components of $ \log|u_\bullet |$ 
for $ \bullet = K, -K , 0 $, respectively, where $ u_\bullet $ spans 
the kernel of $ D ( \alpha ) - \bullet   $ on 
$ L^2_0 ( \mathbb C/\Lambda ; \mathbb C^2 ) $,
and $ \alpha $ is the first 
real magic angle for \eqref{eq:BMU}; $ u_{K}, u_{-K}, u_0  $ vanish at $-z_S ,z_S$ (marked by \red{$\bullet$}), and $ 0 $, respectively. We also indicate (\red{$-$}) the hexagon spanned by 
$ \pm z_S + \Lambda $.
The states $ u_{\pm K } $ exist for all $ \alpha$'s (Proposition \ref{p:protea}) and, 
in the case of a simple $ \alpha \in \mathcal A $ 
have zeros at $ \pm z_S $ (Theorem \ref{t:zS});
 see \url{https://math.berkeley.edu/~zworski/magic.mp4}
for the plot of $ \log | u_{-K} |$ as $ \alpha $ changes. }
\end{figure}
\end{center}

\begin{proof}[Proof of Theorem \ref{t:HtL}] We first note that if $ E_j ( \alpha,  p ) > 0$ for $ j > 1 $
and $\alpha \in A$ then $ E_1 ( \alpha, p ) = 0 $ is a double eigenvalue of $ H_p ( \alpha ) $.
But that means that  $ D ( \alpha ) +  p  $ on $ L^2_{0 } $ is one dimensional (see
the Proof of Theorem \ref{t:1} in \S \ref{s:specH}). 

Hence we need to show that if $ \alpha \in \mathcal A $ and
$ \dim \ker_{L^2_{ p } ( \mathbb C/\Gamma; \mathbb C^2)} D ( \alpha ) = 1 $, for a {\em fixed} 
$p \in \mathbb C $ then
$ E_j ( \alpha,  k ) > 0$ for all $  k $ and $ j > 1 $. To do that we proceed by contradiction 
and suppose that
there exists $  k $ such that $ E_1 ( \alpha,  k ) = E_2 ( \alpha, 
 k) $.

First consider the easy case of $  k =  p $ and we have two 
independent $  v_j $, $ j =1,2 $ in $ \ker_{L^2_{ 0} } ( D ( \alpha ) + p  ) $.
Then $ \widetilde { v}_j ( z ) = \tau(p)  v_j ( z ) $ satisfy 
$ D ( \alpha ) \widetilde { v}_j = 0 $ and $  v_j \in L^2_{ k}=L^2_{ p} $, which gives the desired contradiction.

Now assume that $  k \neq   p$.  Propositions \ref{p:protea} and \ref{p:vanish} give a nontrivial $  u_{ K } \in L^2_{ 0 } $ such that 
$  u_{ K} ( \varepsilon z_S ) = 0 $ where $ \varepsilon \in \{ \pm 1 \} $. 
Put $  z_0 := \varepsilon z_S $,  $ u_{K}  ( z_0 ) = 0 $. 
Using \eqref{eq:delta} we define 
%\blue{If we change the definition of $ u_K $ should we change $ K $ to $ - K$ below?}
\begin{equation}
\label{eq:vz}   v ( z ) := F_{ k - K  } ( z - z_0)  u_{ K} , 
\ \  v   \in L^2_{ 0} ( \mathbb C/\Gamma; 
\mathbb C^2 ) ,  \ \  ( D ( \alpha ) +  k )  v= 0 .
 \end{equation}
  Since 
 $E_{2} ( \alpha,  k  ) = 0 $,  there exists $  w \in L^2_{ 0} $, 
 independent of $  v $ and such that $ ( D ( \alpha) +  k )   w = 0 $. 
If $  v = ( \varphi_1 , \varphi_2 ) $ and $  w = ( \psi_1 , \psi_2 ) $, 
 we form the Wronskian $ W := \varphi_{1} \psi_{2} - \varphi_{2} \psi_{1}$ which 
satisfies
\begin{equation}
\label{eq:Wr0}  ( 2 D_{\bar z } + 2  k ) W = 0 , \ \ \   
 W ( z +  \gamma ) = W ( z )  , \ \ \gamma \in \Lambda .  \end{equation}
(Since $ \mathscr L_\gamma u = u$, $ \varphi_1 ( z+ \gamma ) = e^{-i \langle \gamma, K \rangle }
\varphi_1 ( z ) $ and $ \varphi_2 ( z + \gamma ) = e^{i \langle \gamma, K \rangle }
\varphi_2 ( z ) $, and similarly for $ \psi_1$ and  $ \psi_2 $. That shows periodicity of 
$ W $.)
The definition of $ F_{ k -K } $ in \eqref{eq:defFk}
shows that 
\begin{equation}
\label{eq:Fk0}  F_{ k - K } ( z_1 - z_0) = 0 , \ \ \ z_1 := z_0 + z(k-K ) ,
\end{equation}
so that \eqref{eq:vz} gives 
$  v ( z_1 ) = 0 $. 
This implies that $ W ( z_1 ) = 0 $. If $ 2k \notin \Lambda^* $, $ W \equiv 0 $ 
since $ 2 D_{\bar z } + 2k $ is invertible. Otherwise we note that
$ W ( z ) = e^{-  i \langle 2k , z - z_1 \rangle } W ( z_1 ) = 0 $. 

Since $ W = 0 $, 
\begin{equation}   
\label{eq:defg} w ( z ) = g ( z )  v ( z ),  \ \  g \in C^\infty ( \Omega ), \ \ 
\partial_{\bar z } g|_\Omega = 0  , \ \ \  g ( z + \gamma ) = g ( z ) , \ \ z \in\Lambda ,
\end{equation}
{where $ \Omega :=\complement \{ z : v ( z ) = 0 \} $.}
Also $ g \not \equiv 1$ as $  v $ and 
$  w $ are independent. We claim that $ g $ is a meromorphic function on $ \mathbb C/\Lambda$.
For that fix any $ \underline z \in \mathbb C $ and write $ w = (w_1, w_2 )^t $, $ v = ( v_1, v_2 )^t$.
Then $ g = w_1/v_1 = w_2/v_2$, and 
$w_1 ( {\underline z} + \zeta ) = G_1 ( \zeta, \bar \zeta )$, 
$ v_1 ( - {\underline z} - \zeta ) = G_2 ( \zeta , \bar \zeta ) $, where
$ G_j : B _{\CC^2 } ( 0, \delta ) \to \CC $ are holomorphic functions (this follows
from real analyticity of $ w $ and $ v $, which is a consequence of the 
ellipticity of the equation and analyticity of $ U $ -- see \cite[Theorem 8.6.1]{H1}). 
The definition of $ g $ and the fact that $ \partial_{\bar z } g = 0 $
away from zeros of $ v  $ shows
 that $ G_1 ( \zeta, \xi ) = g ( {\underline z} + \zeta ) G_2 ( \zeta , \xi ) $. We can then choose $ \xi_0 $ such that $ G_2 ( \zeta, \xi_0 ) $ is not 
identically zero (if no such $ \xi_0 $ existed, $ v _1 \equiv 0$, and hence, from the equation, $v \equiv 0 $). But then 
$ \zeta \mapsto g ( {\underline z} + \zeta ) = G_1 ( \zeta, \xi_0 )/G_2 ( \zeta, \xi_0 ) $ is meromorphic near $ \zeta = 0 $ and, as $ {\underline z} $ was
arbitrary, everywhere.

The nontrivial meromorphic function $ g $ has to vanish at at some point, say $ z_2 $.
Hence $ w( z_2 ) = 0 $. 
We define $ z_3 $ (unique $ \!\! \!\!\mod \Lambda $
and not congruent to $ z_0 $) so that $ F_{ k -K } ( z_2- z_3 ) = 0 $. 
Hence, the following functions are smooth,
\begin{equation}
\label{eq:wide2w}   \widetilde { v} ( z ) := \frac{ \tau( p )  F_{ p - K }  ( z - z_0)}{
F_{ k - K } ( z - z_0) }   v ( z ) , \ \ \
\widetilde  { w } ( z ) := \frac{ \tau ( p ) F_{ p - K} ( z -z_3 ) }{
F_{ k-K } ( z - z_3  ) }  w ( z ) , \end{equation}
and 
\begin{equation}
\label{eq:vtwt}  D ( \alpha ) \widetilde { v} = 0 , \ \  D ( \alpha ) \widetilde { w} = 0 ,   
\ \ \
\widetilde { v}, \ \widetilde { w} \in L^2_{ p } ( \mathbb C/\Gamma ; 
\mathbb C^2 ) . \end{equation} 
Consequently, since we assumed that 
$ \dim \ker_{L^2_p }  D ( \alpha ) = 1 $,
  there exists $ c_0 \in \mathbb C $ such  $ \widetilde v ( z) = c_0 \widetilde w ( z ) $.
Lemma \eqref{l:zeros} shows that $ \widetilde v ( z ) $ has a unique simple zero. 
Returning to \eqref{eq:wide2w} we conclude that $ z_3 = z_0$ and that 
$ v ( z ) = c_0 w ( z ) $. In other words, $ \dim \ker_{L^2_k } D ( \alpha ) = 1 $. 
\end{proof} 

%\red{\begin{lemm}
%\label{l:z0S}
%Suppose that $ \dim  \ker_{ L^2_{k,p} } D ( \alpha ) = 1$,  $ k \in \mathcal K $, $ p \in \mathbb Z_3 $.
%Then $ u \in \ker_{ L^2_{k,p} } D ( \alpha ) \setminus \{ 0 \} $ can only vanish at $ \{ 0 , z_S , - z_S\}
%+ \Lambda $.
%\end{lemm}}

%\red{\begin{lemm}
%\label{l:zS}
%Suppose that $ \dim \ker_{L^2_{\pm K, 0 } } D ( \alpha ) = 1 $. Then $ u \in \ker_{ L^2_{ K,0} } D ( \alpha ) \setminus \{ 0 \} $ can only vanish at $ z_S 
%+ \Lambda $.
%\end{lemm}}

\begin{proof}[Proof of Theorem \ref{t:zS}]
We will rely on Lemma \ref{l:van} in several places.
We write $  u ( z ) := \tau ( - K ) u_{-K}  = ( \psi_1 (z) , \psi_2 ( z) ) \in \ker_{L^2_{-K,0} 
 ( \mathbb C /\Lambda; \mathbb C^2 ) } D ( \alpha )  $ 
% \blue{It seems to me that
% this means that we want $ u_K \in \ker_{ H^1_0 } ( D ( \alpha) + K ) $. Perhaps I am just
% losing it...}
  and
assume that $  u ( z_0 ) = 0 $. We recall from Proposition \ref{p:protea} that 
$  u $ has to vanish at $ z_S $ or at $ - z_S $.

We first show that $ z_0 = \pm z_S $. Suppose otherwise and that,  in addition,   $ z_0 \neq 0 $. In that case, $ \omega^j z_0 $ are
three distinct points on $ \CC/\Lambda $ adding up to $ 0 $. Hence,
there exists a $ \Lambda$-periodic meromorphic function $ g_{z_0} $ with simple poles at
$ \omega^j z_0 + \Lambda $ which satisfies
$ g_{ z_0 } ( \omega z ) = g_{ z_0 } ( z ) $.  
This is a general fact (see \cite[\S I.6]{tata}) and we can %for instance 
 take %(see also \eqref{eq:deff} for a similar calculation)
 \[  g_{z_0 } ( z ) = c \prod_{ j=0}^2 \frac{ \theta( z \bar \omega^{j} + z_0 )  }{ \theta (  z \bar \omega^{j}  - z_0  ) }, \]
But this means that $ \widetilde { u} ( z ) := g_{z_0} ( z )   u ( z) $ satisfies $ D ( \alpha ) \widetilde { u } = 0 $
(see Lemma \ref{l:van})
and $ \widetilde { u } \in L^2_{ -K, 0 } ( \CC/\Gamma ) $, $ \widetilde u \not{\! \parallel} \, \, u $. 
Since we assumed simplicity, this is impossible.

We now need to eliminate the possibility that $ z_0 = 0 $. 
From the vanishing of the Wronskian \eqref{eq:Wron} ($ \alpha \in \mathcal A$) and \eqref{eq:+2-}  and we
see that
$ \mathscr E  u ( z) = f ( z )  u ( z)  $
where
\begin{equation}
\label{eq:deff} 
f ( z ) := \frac{ \psi_2 ( -z ) }{\psi_1 ( z ) }. 
\end{equation}
This function is satisfies
\begin{equation}
\label{eq:propf} 
f ( z +  \gamma ) = e^{ -  i \langle \gamma , K \rangle }  f ( z ) ,  \  \ \gamma \in \Lambda , \ \  f ( \omega z ) = f ( z ) , \ \ 
f ( z ) f ( -z ) = -1 .%, \\
\end{equation}

In fact, the holomorphy away from the zeros of $ \psi_1 $ follows from calculating 
$ D_{\bar z } f $ the equations for $ \psi_j $ and the vanishing of the Wronskian \eqref{eq:Wron}. The latter also shows the
functional equation for $ z \mapsto - z $ and from 
the fact that $  u \in L^2_{-K,0} ( \CC/\Lambda ) $ we deduce
quasi-periodicity and invariance under $ z \mapsto \omega z $. We also see that $ f $ is meromorphic 
using the same argument as in the proof of Theorem \ref{t:HtL} (see \eqref{eq:defg}). 
%\red{(to be commented out?) 
%\[  \begin{split} 2 D_{\bar z } f & = - \psi_1 ( z )^{-2} ( ( 2 D_{\bar z}  \psi_2( - z) \psi_1 ( z ) + 
%\psi_2 ( - z ) ( 2 D_{\bar z } \psi_1 ) ( z ) ) \\
%& = U (z)  \psi_1 ( z )^{-2} (  \psi_1( - z) \psi_1 ( z ) + \psi_2 ( - z )   \psi_2   ( z ) ) \\ & = U (z)
% \psi_1 ( z )^{-2} W ( u, \mathscr E u ) = 0. \end{split}  \]
% We have 
% \[  e^{ i \langle \gamma, K \rangle } \psi_1 ( z + \gamma ) = e^{ i \langle \gamma, K \rangle }
% \psi_1 ( z  ) , \ \
% e^{ i \langle \gamma, K \rangle } \psi_2 ( -z - \gamma ) = e^{ - i \langle \gamma , K \rangle } 
% \psi_2 ( -z ) , \]
% which show the quasiperiodicity of $ f $ ($ 2 K \equiv -K \mod \Lambda^* $).  The functional
% equation for $ z \mapsto - z $ follows from the vanishing of the Wronskian: 
% $ \psi_2 ( z ) / \psi_1 ( z ) = - \psi_1 ( -z ) / \psi_2 ( z ) $.
% }
In particular, the functional equation shows that $ f $ is regular 
at $ 0 $. 

With this in place, we now show that a zero at $ z_0 = 0 $ is impossible. 
We claim that 
\begin{equation}
\label{eq:023} \psi_1 ( 0 )  = 0 \ \Longrightarrow \  \partial_{z}^k \partial_{\bar z }^{\ell} \psi_1 ( 0 )  = 0 , \ 
 k \leq 2 , \ \ell \geq 0.  
\end{equation}
This implies that if $\psi_1 ( 0 ) = 0 $ then $ \psi_j ( z ) = 
z^3 \widetilde \psi_j ( z) $. But this means that
\begin{equation}
\label{eq:wp}  \widetilde { u } ( z ) := \wp' ( z ; \omega,
1  )  u \in   L^2_{ -K, 0 } ( \CC/\Gamma_3 ) , \end{equation}
and satisfies 
$ D ( \alpha ) \widetilde { u } = 0 $. Projective uniqueness of $  u $
(uniqueness up to a multiplicative constant) shows that this is impossible.
(Here $ \wp ( z; \omega_1 , \omega_2 ) $ is the Weierstrass 
$\wp$-function -- see \cite[\S I.6]{tata}. It is periodic with respect to $ \ZZ \omega_1 + \ZZ \omega_2 $ and its derivative has a pole of order 
$ 3 $ at $ z = 0 $. For $ \omega = e^{ 2 \pi i/3 } $ we also have 
$  \wp' (\omega  z ; \omega,
1  ) =  \wp' ( z ; \omega,
1  )  $.)

To prove \eqref{eq:023}
we consider expansions at $ ( z , \bar z ) = ( 0 , 0 ) $:
denoting by $ \equiv $ congruency modulo $3$ and using properties of 
$ U $ and 
$ \psi_1 ( \omega z ) = \psi_1 ( z ) $, we obtain
\begin{equation}
\label{eq:defakl}
 \psi_1 ( z ) = \sum_{ k \equiv \ell}  a_{k\ell } z^k \bar z^\ell , \ \
U ( z ) = \sum_{ p   \equiv q+ 1 } b_{pq} z^p \bar z^q , \ \ 
f (- z ) = \sum_{ k \equiv 0 } f_{k} z^k .
\end{equation}
The equation $ 2 D_{\bar z } \psi_1 ( z ) +\alpha U ( z ) f ( - z ) \psi_1 ( -z )
= 0 $ then becomes 
\[ \sum_{ k \equiv \ell }\left[ (2\ell/i ) a_{k\ell} z^k  \bar z^{\ell-1} 
+\alpha (-1)^k
\sum_{  p \geq 1 } \sum_{ q \geq 0 } \sum_{ r \geq 0 } b_{pq} f_r a_{ k\ell} z^{ k+r+p}\bar z^{ \ell+q }   \right] = 0\]
The vanishing of the coefficients of 
 $ z^k \bar z^\ell $ then gives (with the convention that $ a_{k\ell} = 0 $
 for $ k < 0 $ or $ \ell < 0 $)
\begin{equation}
\label{eq:propakl} 
a_{k,\ell+1}  
=  \sum_{ r \leq k-1 , s \leq \ell } g_{rs}^{k\ell} a_{rs} , 
\end{equation}
where $g_{rs}^{k\ell}$ are some constant depending on $k,l,r$ and $s$.
By assumption $ a_{00} = 0 $ and from \eqref{eq:defakl}, $ a_{10} = 
a_{20} = 0 $. Hence, \eqref{eq:propakl} shows that $ a_{k\ell} = 0 $
for $ k \leq 2 $ and all $ \ell$, proving \eqref{eq:023}.

Hence, $  u (z_0 ) = 0 $ implies that $ z_0 = \pm  z_S $. 
 We now see that the zero can occur at only one of the two points. Indeed, if $u$ vanishes at both $-z_S$ and $z_S$ then (note that $ z (K ) = - z_S $)
 \begin{equation}
 \label{eq:pmzS}  \begin{split} \widetilde u ( z ) & := F_{K} ( z - z_S ) F_{- K} ( z + z_S ) u ( z ) 
 \\ & = 
 e^{- i ( z_S - \bar z_S ) K } \frac{\theta ( z )^2 }{ \theta ( z - z_S ) \theta ( z + z_S ) } u ( z) \in
 \ker_{ L^2_{ -K} (\mathbb C /\Lambda; \mathbb C^2) } D ( \alpha ) , \end{split} \end{equation}
 and $ \widetilde u  \not{\! \parallel} \, \,  u $. But this contradicts simplicity.

To show that $ u $ has to vanish at $ z_S $ we analyse $ f $ (defined in \eqref{eq:deff}) near $ \pm z_S $. 
From \eqref{eq:propf} and the fact that 
$   \omega z_S   %= \omega z ( K ) = \omega (\sqrt 3 k/ 4 \pi i) 4 \pi / 3 = \omega( - \frac13 - \frac 23
%\omega ) =  -\frac13 \omega + \frac 23 ( \omega + 1 ) = - \frac13 - \frac 23 \omega + ( 1 + \omega )
= z_S -  ( 1  +\omega )  $, 
we obtain (see \eqref{eq:propf}), 
\begin{equation*}   f (  z_S + \zeta ) = f (  \omega z_S + \omega \zeta ) = 
f ( z_S -  1 - \omega - \omega \zeta ) = %e^{  i \langle 1 + \omega , K \rangle } 
 %f ( -z_S + \omega \zeta ) =
 \omega f ( z_S  + \omega \zeta ), \end{equation*}
that is $ f ( z_S + \omega \zeta ) =\bar \omega f ( z_S + \zeta ) $,
and, in view of the functional equation in \eqref{eq:propf}, 
%this implies that 
$ f ( -z_S + \omega \zeta ) %=  - f ( z_S - \omega \zeta )^{-1}  =  - \omega f ( z_S - \zeta )^{-1}  =
=\omega f ( - z_S + \zeta ) $. 
Hence, there exists $k_0 \in \mathbb Z $ such that
\begin{equation}
\label{eq:expf} 
f(-z_S+ \zeta)=\sum_{k\geqslant -k_0}\zeta^{-2+3k}f_k,\quad f(z_S -  \zeta)=\sum_{\ell \geqslant k_0}\zeta^{2+3\ell }g_\ell, \ \ \ f_{-k_0 } g_{k_0 } = -1 . \end{equation}
%\red{ $- (-2 + 3k_0 ) = 2 - 3 k_0 = -1 + 3 ( 1 - k_0 ) $} 
We also note that, \eqref{eq:defzS} and the definition of $ u = \tau ( -K ) u_{-K} = ( \psi_1, \psi_2 )^t $
gives
\begin{equation}
\label{eq:psif} \psi_1 ( z ) =  f ( z )^{-1}  \psi_2 ( -z ) =  - f ( - z ) \psi_2 ( - z ) , \ \ \psi_2 ( \mp z_S ) = 0 . \end{equation}

Suppose that $u(z_S)\neq 0$ (which is equivalent to  $u(-z_S) =  0$). Then \eqref{eq:psif} shows that $f(z)$ has a pole at $- z_S$. The expansion \eqref{eq:expf}  implies that the pole is of order at least $2$. But using \eqref{eq:psif} again,% the relation $\psi_1 (  z ) = - f ( -z ) \psi_2 (  -z )  $ 
shows that $-z_S$ is a zero of order at least $2$ of the function $\psi_2$, in the sense that
$ \psi_2 ( -z_S + \zeta ) = \zeta^2 \widetilde \Psi_2 ( \zeta)  $, $ \Psi \in C^\omega $, near 
$ \zeta = 0 $.
 Moreover, \eqref{eq:propf} shows that $z_S$ is a zero of order at least $2$ of $f$ and from $\psi_1 (  -z ) = - f (  z ) \psi_2 (  z )  $, we deduce that $ \psi_1 ( - z_S + \zeta ) = \zeta^2 \Psi_1 ( \zeta ) $,
 $ \Psi_1 \in C^\omega $.

We have therefore proved that $ \zeta^{-2} u ( -z_S + \zeta ) $ is
smooth near $ 0 $. But this implies that
\begin{equation}
\label{eq:Weier}  \widetilde{  u} ( z ) = \wp ( z + z_S ;  \omega, 1)
  u ( z ) \in L^2_{-K} ( \CC/\Gamma) \end{equation}
which solves $ D ( \alpha ) \widetilde{  u } = 0 $, and $ \widetilde u \not{\! \parallel} \, \, u $, 
a contradiction.
This implies that $ u$ vanishes only at the point $z=z_S$. 

We want to show that $ \partial_z u ( z_S ) \neq 0 $. 
Since $ u \in L^2_{-K, 0} $, we check that 
\[ \begin{gathered} \psi_1 ( z_S + \omega \zeta ) =  
\psi_1 ( z_S + \zeta ), 
\ \ \ 
\psi_2 ( {\omega}z_S + \zeta ) =
\bar \omega \psi_2 ( z_S + \zeta ). \end{gathered} \]
Since $ u ( - z_S ) \neq 0 $ and $ \psi_2 ({-}  z_S ) = 0$, we see that
$ \psi_1 ( - z_S ) \neq 0 $. We conclude from \eqref{eq:psif} 
that  $ f $ vanishes at $ - z_S$, so that 
using \eqref{eq:expf},  $ f ( -z_S {-} \zeta ) = \sum_{ k \geq 0 } F_k \zeta^{ 1 + 3k } $.
%Lemma \ref{l:van} shows that $ \partial_{\bar z}^\ell  u ( z_S )
%= 0 $ for all $ \ell$. To see that $ \partial_z  u ( z_S ) \neq 0 $,
%we write $ f (  z_S - \zeta ) = \sum_{ k \geq 0 } f_k \zeta^{ 1 + 3 k } $
%(since $ \psi_1 ( -z_S ) \neq 0 $, $ \psi_2 ( z_S ) = 0 $, $f $ has to have a pole
%at $ z_S $
%Since
%\[ \psi_2 (  - z_S ) = 0 \ \text{ and } \ \partial_{\bar z } \psi_2 (  z_S) 
%= -\alpha \tfrac i 2  
%U ( z_S ) \psi_1 ( - z_S ) = -\alpha\tfrac{3i} 2 \psi_1 (  z_S ) =: \eta \neq 0 , \]
We then have
\[ %\begin{split} \psi_1 ( z_S + \zeta )  &= -  f (-  z_S - \zeta ) \psi_2 (  - z_S - \zeta ) 
%= - \sum_{ k \geq 1 } f_k \zeta^{1 + 3k } ( \eta \bar \zeta + 
%\beta \zeta + \mathcal O ( |\zeta|^2 ) ) , \\
\psi_2 ( z_S + \zeta )  = f (  - z_S - \zeta ) \psi_1 (  - z_S - \zeta ) =
\left(\sum_{ k \geq 0 } F_k \zeta^{1 + 3k }\right) ( \gamma + \mathcal O ( |\zeta|) )
 , \ \ \gamma := \psi_1 ( - z_S ) \neq 0 .
%\end{split}
 \] 
We conclude that if $ \partial_z  u ( z_S ) = \partial_\zeta u ( z_S + \zeta ) |_{\zeta =0 } = 0$ then 
$ F_0  = 0 $. Since we also have
\[ \psi_1 ( z_S + \zeta ) = -  f (-  z_S - \zeta ) \psi_2 (  - z_S - \zeta ) 
= - \left(\sum_{ k \geq 0 } F_k \zeta^{1 + 3k }   \right)\psi_2 (  - z_S - \zeta ) , \]
we conclude that  $  \zeta^{-3}   u ( z_S + \zeta ) $ is smooth near $ 0 $.
But this gives a contradiction as in \eqref{eq:wp}. 

The final conclusion \eqref{eq:t3} follows from \eqref{eq:K2k} applied with $ k'-K = k $, $ z_0 = z_S $, 
and the fact that $ F_{k'}  ( z - z_S ) $ vanishes simply and uniquely at $ z_S + z( k' ) = z( k ) = 
\sqrt 3 k / 4 \pi i + \Lambda $ (see \eqref{eq:defFk}). 
\end{proof}

\noindent
{\bf Remark.} Rather than considering the vanishing of
 $ \tau ( -K ) u_{-K} \in L^2_{ -K , 0 } $ we could look at $ u_0 $, $ \ker_{L^2_{0} } ( D ( \alpha )) =
 \mathbb C u_0 $, $ \alpha \in \mathcal A $. One can easily show (see \cite[Proposition 3.6]{dynip}) 
 that $ u_0 \in L^2_{ 0 , 2} $ and that $ \mathscr E u_0 = \pm i u_0 $ (see \eqref{eq:symmE}
 and note that $ \Spec_{L^2_{0} } ( \mathscr E ) =  \{ i, -i \} $). That implies that $ u_0 $ vanishes
 at zero and that other zeros are symmetric with respect to the origin. But $ 0$ is the only 
 zero as the
 same  argument as in \eqref{eq:pmzS} would contradict simplicity. Since, again by 
 simplicity, $ \tau ( -K ) u_{-K} = c_0 \tau ( -K )  F_{-K} ( z ) u_0 $,  that gives 
 a different (and perhaps simpler) proof that $ \tau ( -K ) u_{-K} $ vanishes only at $ z_S $. 
 
 As suggested by Mengxuan Yang, we can then see directly that the $ u_0 $ vanishes
 simply at $ 0 $ (which then implies that $ u_{-K } $ vanishes simply at $ z_S $). For that
 consider $ u_1 \in C^\infty ( \mathbb C ) $ such that $ u_0 ( z ) =  z u_1 ( z ) $ 
 (this follows from Lemma \ref{l:van}). But then $  D ( \alpha ) ( z u_1 ( z ) ) = 
 z D ( \alpha ) u_1 ( z) = 0 $, and as $ u_1 $ is smooth, $ D ( \alpha ) u_1 ( z ) = 0 $ for
 $ z \in \mathbb C $. Hence, if  $ u_1 ( 0 ) = 0 $ then Lemma \ref{l:van} shows that $ u ( z ) = z
 u_1 ( z ) = z^2 u_2 ( z ) $, $ u_2 \in C^\infty ( \mathbb C ) $, and the $ \wp $-function argument (see \eqref{eq:Weier}) contradicts
 simplicity in $ L^2_0 $. 

 We opted for
 a direct discussion of $ \tau (K ) u_K $ (and the proof of simplicity of the zero) 
 as that protected state which exist for all $ \alpha $ and its zero were central in the
 original physics presentation \cite{magic}.

\section{Theorem \ref{theo:Chern}  and two numerical observations}
\label{s:chern}

Here we present two numerical observations about the structure of flat bands and 
compute the Chern number of the flat band.

\subsection{Fixed shape of the rescaled flat band}
We define 
rescaled bands as follows:
\begin{equation}
\label{eq:defhatE} \widehat E_j ( \alpha, k ) :=   \frac{ E_j ( \alpha, k )}
{ \max_{k}  E_1 ( \alpha, k ) } , \ \ \alpha \in \mathbb R . 
\end{equation}
and notice that for $ \alpha $ near $ \alpha$'s near elements of $ \mathcal A_{\mathbb R }$, 
\begin{equation}
\label{eq:E2U}
E_1 ( \alpha, k ) \simeq | U ( z ( k ) ) | , \ \ \  z ( k ) :=  \frac{ \sqrt 3 k }{ 4 \pi i } k : \Lambda^* \to \Lambda , 
\end{equation} 
see Figure \ref{f:pbands} and, for an animated version 
\url{https://math.berkeley.edu/~zworski/KKmovie.mp4}. We note that the $ |U ( z (k)  ) |$ is the
simplest function with symmetries of $ E_1 ( k ) $ and conic singularities at $ \pm K $.
\begin{figure}
\includegraphics[width=11cm]{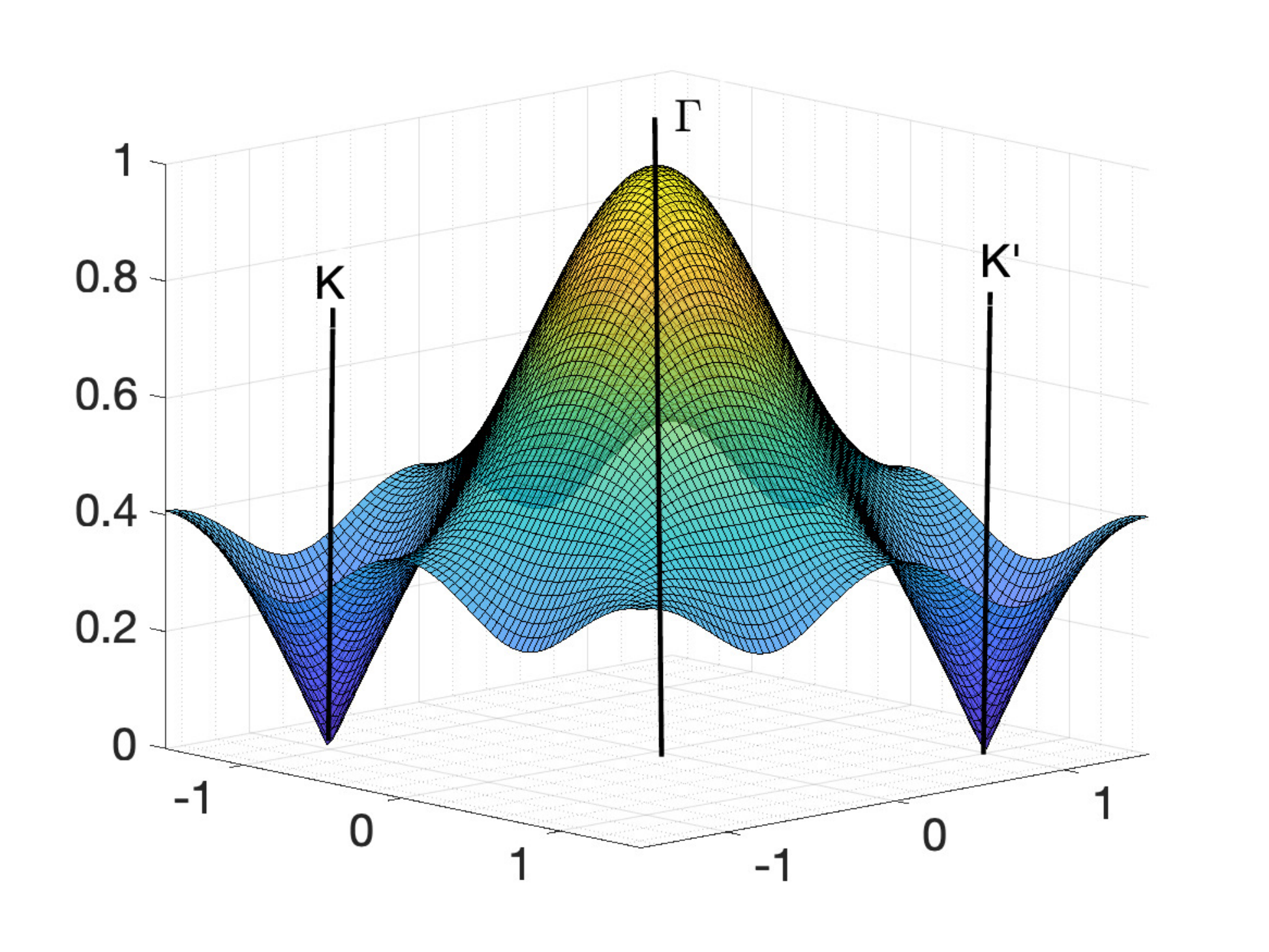}
\caption{\label{f:dE} Normalized $ \partial_\alpha E_1 ( \widetilde \alpha, k )$ at first magic angle. The protected zero energy states at the $K$ and $K'$ points are preserved, the maximum is attained at the $\Gamma$-point.}
\end{figure}

The following heuristic explanation was suggested by Ledwith et al \cite{priv}. Assuming
that the flat band is simple consider perturbation theory of $ H_{k } ( \alpha ) $
near $ \widetilde \alpha\in \mathcal A_{\mathbb R } $:
\begin{equation}
\label{eq:priv} 
\begin{gathered} \partial_\alpha E_1 ( \widetilde \alpha, k ) =
\frac{ | \langle
{V}
 u_{k } , v_{\bar{k } }  \rangle |}
{ \|  u_{k } \| \| v_{\bar{k } } \| }  , \ \ \ 
u_{k } , v_{ {k } } \in L^2_{0 } , \ \ V 
 (z)  := \begin{pmatrix} 0 & U ( z ) \\
U (- z) & 0 \end{pmatrix}, 
\\
( D ( \widetilde \alpha ) - k ) u_{k } = 0 , \ \ \ 
( D ( \widetilde \alpha )^* - \bar {k} ) v_{ {k } } = 0 ,
\end{gathered} 
\end{equation}
where using  \eqref{eq:H2Q} we can take $ v_{ k } := Q u_k $.  
A numerically evaluated graph of $ k \mapsto \partial_\alpha E_1 ( \widetilde \alpha, 
k ) $ is shown in Figure \ref{f:dE}. 
Since $ u_{k}  = F_k u_0 $, $ \ker_{ L^2_0 }  D( \alpha )  =  u_0 $, 
the theta function factors act in some sense as an FBI/Bargmann transform 
(see \cite[Chapter 13]{semi}). A (very formal) application of stationary phase method could then 
reproduce the potential $ U $.

\subsection{The Chern connection and curvature}
\label{s:ccc}
The second numerical observation concerns the behaviour of the curvature of 
a connection on the natural hermitian bundle associated to the flat band.
Since the bundle is holomorphic we use the Chern connection but, as is always
the case, the resulting curvature is the same as the Berry curvature -- see 
\eqref{eq:C2B} for a direct verification in our case.

The numerical observation is shown in Figure \ref{f:curv} (a three dimensional plot of the curvature for one
magic angle) and Figure \ref{f:Ukraine} (the two dimensional plots for the first magic angles). 
We note that the absolute maximum appears at the $ \Gamma $ point, that is
the center of the $k$-space hexagon spanned by translates of $ K $ and $ K' $ (equal to 
$ - K $ in our coordinates), and the minima at $ K $ and $ K'$ -- the 
vertices of the hexagon (the Dirac points). This is supposed to correspond to the fact 
that the bands are closest at $ \Gamma $ and farthest apart at the Dirac points
(see the movie linked to Figure \ref{f:pbands}). So far, we only show that $ \Gamma $, $ K $ and
$K' $ (that is $ \mathcal K $ -- see \eqref{eq:defK}) are critical 
points for the curvature which follows from Proposition \ref{p:H} below.

To describe the objects involved we need to define the hermitian holomorphic
line bundle associated to the flat band. For general definitions and basic facts we
refer to the self-contained appendix. 

\begin{figure}
\includegraphics[width=10cm]{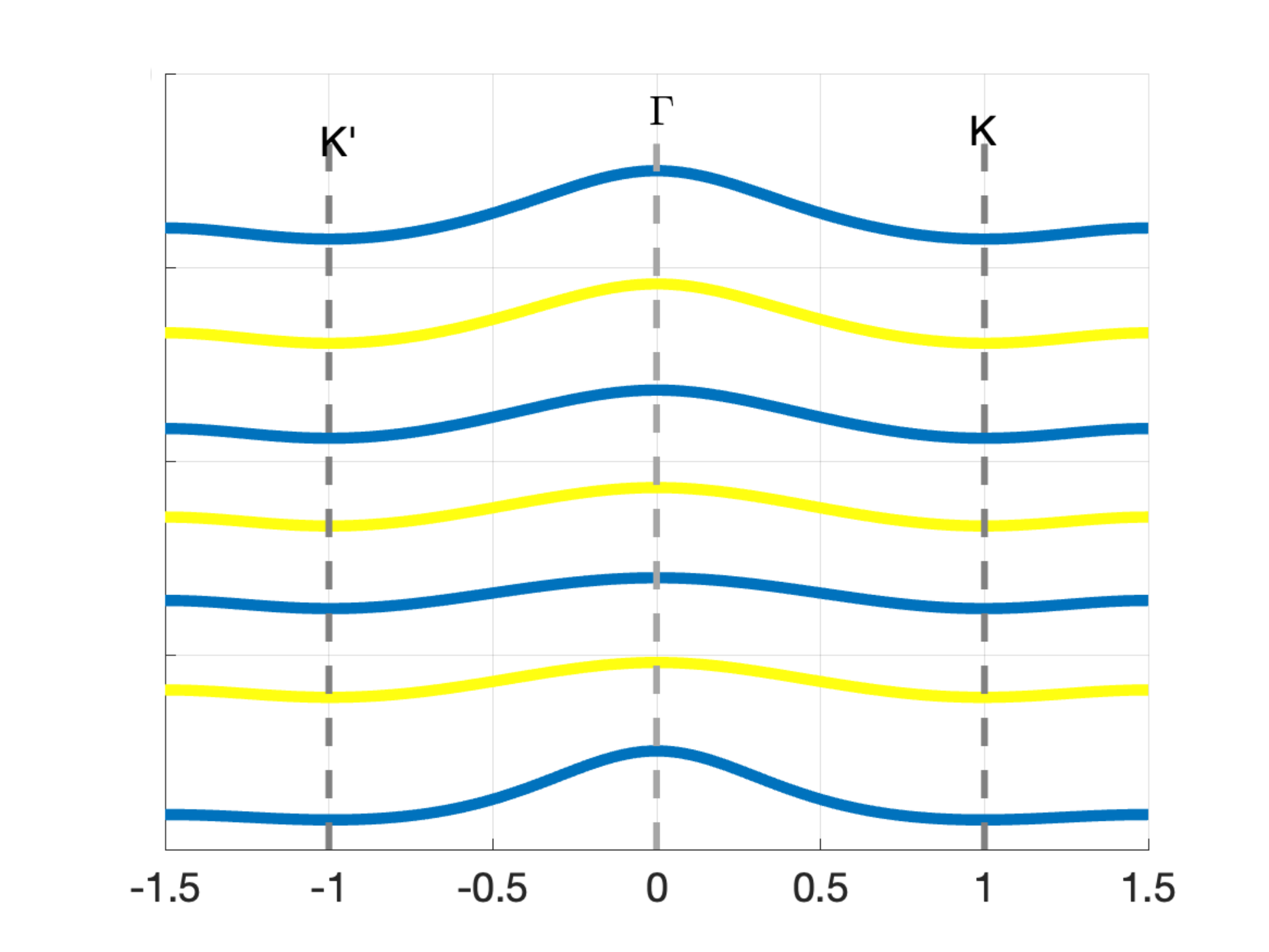}
\caption{\label{f:Ukraine} Cross-section of curvature for $k_x=0$ for the first seven magic angles in increasing order. The extrema at $K,\Gamma,K'$ follow from Prop. \ref{p:H} and the subsequent discussion.}
\end{figure}

We assume that $ \alpha_0 \in \mathcal A $ is simple in the sense of
Theorem \ref{t:zS}: $ \dim \ker_{L^2_{  0 } } ( D ( \alpha_0 ) + k ) = 1 $ for all $ k \in \mathbb C $.  We remark that in \cite[Theorem 3]{bhz1} we established
simplicity of the first real magic $ \alpha$ for the potential used in \cite{magic}. 
Numerical calculations suggest that all real magic $ \alpha$'s for that potential
are indeed simple.

We recall from 
\S \ref{s:theta} (see \eqref{eq:K2k}) that 
$$ \ker_{L^2_0 }  ( D ( \alpha_0 ) + k ) = \mathbb C F_{k-K} (  z -  z_S  ) u_K  (  z ). $$
We then put
\begin{equation}
\label{eq:defuk}  [ u ( k ) ] (  z ) =  u ( k ,  z) :=  F_{k-K} (  z -  z_S   ) u (  z ) , \
\ \  \mathscr L_\gamma u (k ) = u ( k ) , \ \  \gamma \in \Lambda . 
\end{equation}
We also note that \eqref{eq:Fk2p} implies
\begin{equation}
\label{eq:defep}
\begin{gathered}
u ( k+ p  )  =    e_{ p} ( k)^{-1}  \tau ( p )^{ {-1}} u( k )  . 
\end{gathered} 
\end{equation}

Following the standard construction (see for instance \cite[\S 2.1]{panati}) we define
\begin{equation}
\label{eq:defL}  
\begin{gathered}
L := \left\{ [ k , v ]_\tau  \in ( \mathbb C \times L^2_{  0 } ( \mathbb C/ \Lambda ; \mathbb C^2 ) )/
\sim_\tau : v \in \ker_{L^2_{  0 }  ( \mathbb C/ \Lambda ; \mathbb C^2 )} ( 
D ( \alpha_0 ) + k ) \right\}, \\ 
[ k , v]_\tau = [ k', v']_\tau \ \Longleftrightarrow \
( k , v ) \sim_\tau ( k', v' ) \ \Longleftrightarrow \  \exists \, p \in \Lambda  \ 
k' = k + p , \ \  v' = \tau ( p )^{ {-1}} v.  \end{gathered} \end{equation}
We have 
\begin{lemm}
\label{l:holo}
Definition \eqref{eq:defL} gives a holomorphic line bundle over $ \mathbb C / \Lambda $, 
\[   f: L  \to  \mathbb C/\Lambda , \ \  f : [ k , v ]_\tau \to [k] \in \mathbb C / \Lambda . \]
The corresponding family of multipliers in \eqref{eq:CD} is given by $ k \mapsto e_p ( k ) $.
\end{lemm}
\begin{proof}
The action of the discrete group $ \Lambda $, $ \lambda : ( k , v ) \mapsto ( k + p , \tau ( p ) v ) $
on the (trivial) complex line bundle 
\begin{equation}
\label{eq:wL}  \widetilde L := \{ ( k , \kappa u ( k ) ) : k \in \mathbb C , \ \ \kappa \in \mathbb C \}
\simeq \mathbb C_k \times \mathbb C_\kappa  , \end{equation}
(where $ u( k )  $ is defined in \eqref{eq:defuk})  is free and proper, and the quotient map is given by 
$ \pi_\tau ( k , \kappa u ( k ) ) = [ k , \kappa u( k ) ]_{\tau } $. 
Hence its quotient by that action, $ L $, is a smooth complex manifold of dimension 2.

The map $ (p,k ) \mapsto e_p ( k ) $ satisfies conditions in \eqref{eq:CD},
\[ e_{ p + p'} ( k ) = \frac{ \theta ( z ( k ) ) }{ \theta ( z ( k + p + p' ){)} } 
= \frac{ \theta ( z ( k + p) ) } { \theta ( z( k + p + p' ))  } \frac{ \theta ( z ( k ) ) }{\theta ( 
z ( k + p )  {)}} = e_{p'} ( k + p ) e_p ( k ), \]
and for $ p \in \Lambda $ 
we define $ \varphi_p :  \widetilde L  \to \widetilde L $ as in 
\eqref{eq:CD}: $ \varphi_p ( k , \kappa u ( k )  ) = ( k + p , e_p ( k ) \kappa u ( k )  ) $, 
$ \kappa \in \mathbb C $.  We then have 
$ \pi_\tau ( \varphi_p ( k, \kappa u ( k )  ) ) = \pi_\tau ( k , e_p ( k ) \kappa u ( k )  ) $ and this 
gives $ L $ the structure of a complex line bundle over $ \mathbb C /\Lambda $ \end{proof}

\noindent{\bf Remark.} As is implicit in the above proof, the multiplier $ e_p ( k ) $ is the multiplier
of the antiholomorphic theta line bundle over $ \mathbb C /\Lambda^* $.

The hermitian structure is inherited from $ L^2 ( \mathbb C / \Lambda) $ and the resulting hermitian 
structure on $ \widetilde L $ of \eqref{eq:wL}. In coordinates $ ( k , \kappa ) $ on $ \widetilde L $, 
we get
\[ h ( k ) =    \| u ( k )  \|_{L^2 ( \mathbb C/\Lambda)  }^2 ,\]
where we note that \eqref{eq:defuk} shows that $ u ( k ) $ is well defined on 
$ L^2_0 ( \mathbb C/ \Lambda ) $. This gives us also a hermitian structure on $ L $:
from \eqref{eq:defep}
we see that
\begin{equation}
\label{eq:logh} % \log h ( k + \omega m + n ) = \sqrt 3 \pi m^2 + 2 \pi i ( \bar k - k ) m + \log h ( k ) . 
h ( k  ) =  | e_p ( k ) |^2 h ( k + p  ) , \ \  p \in \mathbb Z \oplus \omega \mathbb Z .
\end{equation}
To $ h $ we associate the Chern connection \eqref{eq:chernc} 
and the curvature $ \Omega $, \eqref{eq:curv}.
The general formula \eqref{eq:c1L} then reproduces the calculation from \cite{led}
(where \eqref{eq:logh} was used directly):
\begin{equation}
 \begin{split}
 \label{eq:Chern}c_1 (L) = \frac{i}{2\pi}  \int_{\mathbb C / \mathbb Z \omega \oplus \mathbb Z } \Omega 
=  -1. \end{split} 
\end{equation}
This proves Theorem \ref{theo:Chern}. 

\noindent 
{\bf Remark.} We should stress that $ k \mapsto u ( k ) $ is {\em not} a holomorphic 
section of $ L$. In fact, as indicated by the Chern number,  the line bundle $ L$ does not have any holomorphic sections. The dual line bundle corresponding to the kernel
$ D( \alpha )^* + \bar k $ has the Chern number equal to $ 1 $, and hence has holomorphic 
sections which can be expressed using theta functions.

We also have an explicit formula for $ \Omega $ in terms of $ u ( k ) $:
\begin{equation}
\label{eq:Om2H}
\begin{gathered}
\Omega :=  H ( k )  d\bar k \wedge dk , \ \ \
d\bar k \wedge dk = 2 i d \Re k \wedge d \Im k , \\
\begin{split}
H(k ) & =  \partial_k \partial_{\bar k } \log h ( k ) \\ & 
= 
 \| u ( k ) \|^{-4} \left( \| u ( k ) \|^2 
 \| \partial_k  u ( k ) \|^2 - | \langle \partial_k u (k) , 
 u (k) \rangle |^2 \right) \geq 0 ,  \end{split} \end{gathered}\end{equation}
 and (unlike $ u ( k ) $) $ H \in C^\infty ( \mathbb C / \Lambda ; \mathbb R ) $.
 We note that this is equivalent to the standard formula for the Berry curvature 
  \cite{bearry} (valid also
in non-holomorphic situations):
 \begin{equation}
 \label{eq:C2B}
 H ( k ) = - \Im \langle \partial_{k_1} \varphi ( k ) , \partial_{k_2 } \varphi ( k ) \rangle_{L^2 ( 
 \mathbb C , \mathbb C/3\Lambda ) } , \ \ \
 \varphi ( k ) := u ( k ) / \| u ( k ) \| .  \end{equation}
(This is a special case of a general fact.) 
We also recall the well known independence of $ H ( k ) $ of the phase of $ \varphi $:
\begin{lemm}
\label{l:phasein}
Suppose that $ \gamma ( k ) \in C^\infty ( \mathbb C; \mathbb R ) $ and that
$ \varphi ( k ) \in C^\infty ( \mathbb C ; L^2_{  0} ( \mathbb C/ \Lambda; \mathbb C^2 )) $,
$ \| \varphi ( k ) \|_{ L^2 ( \mathbb C/ \Lambda; \mathbb C^2 ) }\equiv 1 $.
Then, putting $ k_1 = \Re k $, $ k_2 = \Im k $, 
\begin{equation}
\label{eq:phasein}
\Im \langle \partial_{k_1} \varphi ( k ) , \partial_{k_2 } \varphi ( k ) \rangle =
\Im \langle \partial_{k_1} (e^{ i \gamma ( k ) }  \varphi ( k ) ) , \partial_{k_2 } ( e^{ i \gamma ( k ) } \varphi ( k ) ) \rangle. 
\end{equation}
\end{lemm}
\begin{proof} The difference the two sides in \eqref{eq:phasein} is given by 
\[    \Im \left( - \gamma_{k_1} \gamma_{k_2 } \langle \varphi, \varphi \rangle + 
i \gamma_{k_1} \langle \varphi, \varphi_{k_2} \rangle - i \gamma_{k_2} \langle \varphi_{k_1}, \varphi \rangle \right) = \tfrac12 \left( \gamma_{k_1} \partial_{k_2} \| \varphi \|^2 -
\gamma_{k_2} \partial_{k_1} \| \varphi \|^2 \right) 
 \]
and this vanishes as $ \varphi( k ) $ is $L^2$-normalized. \end{proof}

Simplicity of $ \alpha_0 $ has the following consequence:
\begin{prop}
\label{p:H}
Suppose that $ H $ is given by \eqref{eq:Om2H} with $ u ( k )  $ defined in \eqref{eq:defuk}.
Then 
\begin{equation}
\label{eq:Hs}
H ( \omega k ) = H ( k ) .
\end{equation}
\end{prop}
\begin{proof} 
 Since $ [ D( \alpha ) u ] ( \omega  z ) = \omega D ( \alpha ) [ u ( \omega \bullet ) ] (  z ) $, 
 \[ 0=  [( D ( \alpha ) - k ) u ( k )] ( \omega  z ) =
  \omega ( D ( \alpha ) - \bar \omega k )[ u ( k ) ( \omega \bullet) ] (  z )  , \]
and simplicity shows
 that $ u( \omega k ,  z ) = \rho ( k ) u ( k, \bar \omega  z ) $, and $ \| u ( \omega k ) \| = |\rho ( k ) | \| u ( k )\|$,
$  h ( \omega k ) = | \rho ( k ) |^2  h ( k ) $. In particular, $ | \rho ( k )| > 0 $ and, as a function on 
$ \mathbb C $, $ \rho  ( k ) / | \rho ( k ) | = e^{ i\gamma ( k ) } $ for some $ \gamma \in C^\infty ( 
\mathbb C ; \mathbb R ) $.  The conclusion then follows from Lemma \ref{l:phasein}
and \eqref{eq:C2B}.
\end{proof}
 
This proposition shows that elements of $ \mathcal K  $ (that is, $ \Gamma$, $ K $ and $K' $ -- see
\eqref{eq:defK}) are
are critical points of $ H$: suppose that $ p \in \mathcal K $; then (since $ \omega p \equiv p \!\! \mod \Lambda^* $
and $ H $ is $ \Lambda^* $-periodic), 
\[  H (p + \kappa ) = H ( \omega p + \omega \kappa ) = H ( p + \omega \kappa ) , \]
which implies that $ \partial_k H ( p ) = \omega \partial_k H ( p ) $, $
\partial_{\bar k } H ( p ) = \bar \omega \partial_{\bar k} H ( p ) $, that is that $ d_k H ( p ) = 0 $. 
This provides a partial explanation of Figures \ref{f:curv}  and \ref{f:Ukraine}.

 \smallsection{Acknowledgements} 
We would like to thank 
Eslam Khalaf,  Patrick Ledwith, Allan MacDonald, Daniel Parker, Ashvin Vishwanath, Alex Watson
and Mike Zaletel for helpful discussions and their interest in this work. We are also grateful to 
Mitch Luskin and 
Alex Sobolev for informing us of references \cite{sgog} and \cite{dun} respectively, to Simone Warzel for bringing \cite{F15} to our attention, and especially to Zhongkai Tao for comments on 
earlier versions of this paper, and to Mengxuan Yang for suggesting a simpler proof of Theorem \ref{t:zS}
(see the remark at the end of \S \ref{s:prooft}). 
TH and MZ  were partially 
supported
by the National Science Foundation under the grant DMS-1901462
and by the Simons Foundation under a ``Moir\'e Materials Magic" grant.

\vspace{0.5cm}

\begin{center}
\noindent
{\sc  Appendix A: translation between different conventions}
\end{center}
%\vspace{0.1cm}
\renewcommand{\theequation}{A.\arabic{equation}}
\refstepcounter{section}
\renewcommand{\thesection}{A}
\setcounter{equation}{0}

We compare the coordinates use \eqref{eq:propU} to those in 
\cite{beta}, and implicitly in the physics literature -- \cite{magic}. 
One of the advantages of using the lattice $ \Lambda $ is the more straightforward connection 
with $ \theta $ functions.  

In \cite{beta} we considered the following 
operator built from the potential $ U_0 $: 
\begin{equation}
\label{eq:oldD}  \begin{gathered}  \widetilde D ( \alpha ) := 
\begin{pmatrix}
2 D_{\bar \zeta } & \alpha U_0 ( \zeta ) \\
\alpha U_0 ( - \zeta )& 2 D_{\bar \zeta } \end{pmatrix} , \ \ \ \overline{ U_0 ( \bar \zeta ) } = U_0 ( \zeta ) , \\
 U_0 \left( \zeta + \tfrac{ 4 \pi i } 3 ( a_1 \omega + a_2 \omega^2  ) \right) = 
\bar \omega^{a_1 + a_2 } U_0 ( \zeta ) , \ \ \  U_0 ( \omega \zeta ) = \omega U_0 ( \zeta). \end{gathered} 
\end{equation}
We then have  periodicity with respect to 
\[   \Gamma := 4 \pi i ( \omega \mathbb Z + \omega^2 \mathbb Z ) = 4 \pi i \Lambda \, \]
and twisted periodicity with respect to $ \Gamma / 3 $. The dual lattices are given by 
\[ \Gamma^*:=\tfrac{1}{\sqrt{3}} (\omega \ZZ \oplus \omega^2 \ZZ) = \tfrac 1{\sqrt{3}} {\Lambda} , \ \ 
( \tfrac13 \Gamma )^* = 3 \Gamma^* = \sqrt 3 \Lambda . 
 \]
This means that to switch to (twisted) periodicity with respect to $ \Lambda $ we need a
change of variables:
\begin{equation}
\label{eq:z2zeta}     \zeta = \tfrac 4 3 \pi i z , \ \   \tfrac13  \Gamma = \tfrac 4 3 \pi i \Lambda, \ \ 
3 \Gamma^* = ( \tfrac13 \Gamma)^* = \sqrt 3 \Lambda = \frac{ 3 }{ 4 \pi i }  \Lambda^* 
. \end{equation}
Then 
\begin{equation}
\label{eq:unit} \begin{gathered}  
\widetilde D ( \alpha ) = - \frac{ 3}{ 4 \pi i } 
\begin{pmatrix}
2 D_{\bar z  }  & \alpha  U ( z )  \\
\alpha U ( - z ) & 2 D_{\bar z }  \ \end{pmatrix}    , 
\ \ \ U ( z ) := - \tfrac43  \pi i 
U_0 \left(  \tfrac 43 \pi i z  \right).
\end{gathered} \end{equation}
The twisted periodicity condition in \eqref{eq:oldD} corresponds to the condition 
in \eqref{eq:propU} since 
\[  \bar \omega^{a_1 + a_2} = 
e^{ i \langle a_1 \omega + a_2 \omega^2 , K \rangle } , \ \ \   K
= \tfrac 4 {\sqrt 3 } \pi i (- \tfrac13 - \tfrac23 \omega ) = \tfrac{4 } 3 \pi . \]

Floquet theory using $ \mathscr L_\gamma $ defined in \eqref{eq:FL_ev} is equivalent to 
the Floquet theory based on $ \widetilde  {\mathcal L}_{\mathbf a } $ used in \cite{magic}: for
$ u \in L^2_{\rm{loc}} ( \mathbb C; \mathbb C^2 ) $, 
\begin{equation}
\label{eq:defLa0}  
 \widetilde { \mathscr L}_{\mathbf a } u :=  
\begin{pmatrix} \omega^{a_1 + a_2}  & 0  \\
0 & 1 
\end{pmatrix} u ( \zeta  + \mathbf a )   ,   \ \ \ 
\mathbf a = \tfrac{ 4}3 \pi i (  \omega a_1 + \omega^2 a_2 ) \in \tfrac 13 \Gamma, \ 
a_j \in \mathbb Z .  \end{equation}
The Floquet theory based on $ L^2_0$ defined using $ \widetilde  {\mathcal L}_{\mathbf a } $
gives different values of $k \in 3 \Gamma^*$ for protected states. That is easily seen
by considering the spectrum of $ 2 D_{\bar z } $ on that $ L^2_0$ which is given 
(modulo $ 3 \Gamma^* $) by 
\[  \widetilde K = \frac{ 3}{ 4 \pi i }K = - i , \ \ \widetilde K' = 0 , \]
with the $\widetilde  \Gamma $ point corresponding to $ i $ (see also \cite[Proposition 3.2]{dynip}). 

\vspace{0.5cm}
\begin{center}
\noindent
{\sc  Appendix B: holomorphic line bundles over tori}
\end{center}
%\vspace{0.1cm}
\renewcommand{\theequation}{B.\arabic{equation}}
\refstepcounter{section}
\renewcommand{\thesection}{B}
\setcounter{equation}{0}

Suppose $ \Lambda $ is a lattice $ \mathbb Z \oplus \omega \mathbb Z $, 
$ \Im \omega > 0 $ (for us it will be $ \omega = e^{ 2 \pi i /3 } $). 
A holomorphic line bundle $ L $, $ f: L \to \mathbb C/\Lambda $ can be described
 using a 
pullback by the canonical projection $ \pi : \mathbb C \to \mathbb C/ \Lambda $,  that is a (trivial) line bundle
$ \pi^* L $ over $ \mathbb C$ for which the following diagram commutes:
\[ \begin{CD}
\pi^* L    @>>>   L \\
@VVV      @VV{f }V \\
{\mathbb C}    @>\pi>> {\mathbb C}/ \Lambda 
\end{CD}  \]
We can identify 
$ \pi^* L $ with $ \mathbb C \times \mathbb C $ and write its elements as $ ( z, \zeta ) $.
Every line bundle over $ \mathbb C/\Lambda $ is associated to an entire non vanishing function
$ z \mapsto e_\lambda ( z ) $ such that
\begin{equation}
\label{eq:CD}
\begin{gathered} \varphi_\lambda ( z, \zeta ) = ( z + \lambda , e_\lambda ( z ) \zeta ) , \ \ 
e_{\lambda + \lambda' } ( z ) = e_{ \lambda' } ( z + \lambda ) e_{\lambda } ( z) , \ \ 
\lambda, \lambda' \in \Lambda, \\
\begin{CD}
\pi^* L    @>{\varphi_\lambda}>>  \pi^* L \\
@VVV      @VVV \\
L     @>{\operatorname{id}}>> L 
\end{CD}
\end{gathered} \end{equation}
In other words $ L$ is the set of equivalence classes, $ [ ( z , \zeta ) ]_{\Lambda} $ where
\[ ( z , \zeta ) \sim ( z', \zeta' ) \ \Longleftarrow \ \exists\,  \lambda \in \Lambda \ \
( z, \zeta ) = \varphi_{\lambda } ( z', \zeta' ) . \]
We then have 
\begin{equation}
\label{eq:secti} C^\infty ( \mathbb C / \Lambda ; L ) \simeq 
\left\{ u \in C^\infty ( \mathbb C ) : \forall \,  \lambda \in \Lambda,  \ u ( z + \lambda ) = 
e_\lambda ( z ) u  ( z )  \right\}. \end{equation}
Holomorphic sections are defined by replacing $ C^\infty ( \mathbb C )  $ with $ \mathscr O ( \mathbb C)$,  the space of entire functions on $ \mathbb C $.

The functions $ e_\lambda ( z ) $ are not unique: if $ g \in \mathscr O ( \mathbb C) $ then
$ \widetilde e_\lambda ( z ) := e^{ g ( z + \lambda ) } e_\lambda ( z ) e^{ - g ( z ) } $ gives the same
line bundle. 
%To see what functions are allowed consider
%\[   e_\lambda ( z ) := e^{ 2 \pi i q ( z, \lambda ) }, \ \ \ 
%q ( z, \lambda + \lambda' ) - q ( z + \lambda ,\lambda' ) - q ( z, \lambda ) \in \mathbb Z , \ \ 
%z \in\mathbb C, \ \  \lambda , \lambda' \in \Lambda. \]

\noindent
{\bf Remark.} The Appell--Humbert theorem completely characterizes the allowed 
functions $ e_\lambda ( z ) $. Here we will concentrate on the specific $ e_\lambda ( z ) $
arising from the eigenfunctions. 

\subsection{Hermitian structure and the Chern connection}
Hermitian structure provides a notion of length on the fibers of $ L$, $ p^{-1} ( z ) $ locally
described by (with $ |\zeta |^2 = \bar \zeta \zeta $, $ \zeta \in \mathbb C $), 
\begin{equation}
\label{eq:herm}  
\begin{gathered}   \| [ ( z, \zeta )]  \|^2 = h ( z ) | \zeta |^2 , \\
 \| [ ( z, \zeta )] \|^2 = \| [\varphi_\lambda ( z, \zeta ) ] \|^2 \ \Longleftrightarrow \ 
 h ( z )  = h( z + \lambda ) | e_\lambda ( z ) |^2 . \end{gathered}
 \end{equation}
 Conversely any positive smooth function $ h ( z ) $ satisfying the condition in \eqref{eq:herm}
 defines a hermitian metric on $ L $. 
 
Connections on $ L$ are identified with connections on $ \pi^* L \simeq \mathbb C \times \mathbb C $.
The latter are given by $ \eta \in C^\infty ( \mathbb C, T^* \mathbb C ) $ so that we can define the 
actual connection:
\[ D_\eta s = ds + s \eta \in C^\infty ( \mathbb C , T^* \mathbb C)   , \ \  s \in C^\infty ( \mathbb C , \mathbb C ) . \]
This gives a connection on $ L$ provided that 
\[  d ( s ( z + \lambda ) ) + \eta ( z + \lambda ) s ( z + \lambda ) = 
e_\lambda ( z ) ( d s ( z ) + \eta ( z ) s ( z ) ) , \]
that is when
\begin{equation}
\label{eq:conn}
\eta ( z + \lambda ) =  \eta ( z ) -  e_\lambda ( z ) ^{-1}  e'_\lambda ( z ) d z. 
\end{equation}
The {\em Chern connection} is defined by 
\begin{equation}
\label{eq:chernc}  \eta ( z ) = \partial ( \log h ( z ) ) = h ( z )^{-1}  \partial_z h ( z ) dz , 
\end{equation}
(here we denote by $ \partial f = \partial_z f ( z ) dz $, the $( 1,0) $-differential)
and we easily check \eqref{eq:conn} using \eqref{eq:herm} and the holomorphy of
$ z \mapsto e_\lambda ( z ) $:
\[   \begin{split}
\eta ( z + \lambda ) & =  \partial_z \log h ( z + \lambda ) dz  = 
\partial_z ( - \log ( e_\lambda ( z )  \overline{ e_\lambda ( z ) } ) + \log h( z ) ) dz  \\
& =
\eta ( z ) - e_\lambda ( z ) ^{-1}  e'_\lambda ( z ) d z. \end{split} \]

\subsection{Curvature and Chern numbers}

In this simplest case the curvature is just the differential of $ \eta $ and it is 
a well defined (unlike $ \eta $) $(1,1)$-differential form on $ \mathbb C /\Lambda $:
\begin{equation}
\label{eq:curv}
 \Omega := \bar \partial \eta = \bar \partial\partial ( \log h ( z ) ) = 
\partial_{\bar z } \partial_z  (\log h ( z )) d\bar z \wedge d z.   
\end{equation}
Indeed, the holomorphy of $ z \mapsto e_\lambda ( z ) $ gives
\[ \partial_{\bar z } \partial_z \log h ( z + \lambda ) = 
\partial_{\bar z } \left( \partial_z \log h ( z) - e_\lambda ( z ) ^{-1}  e'_\lambda ( z ) \right) = 
\partial_{\bar z} \partial_z \log h ( z ) . \]

The Chern number (since we are in complex dimension one) is defined as
\begin{equation}
\label{eq:Chern}  c_1 ( L ) := \frac{i}{ 2 \pi } \int_{ \mathbb C/\Lambda } \Omega \in \mathbb Z . 
\end{equation}
To see that $ c_1 ( L ) $ is an integer, we choose a fundamental domain of $ \Lambda $, 
$ F $, and apply Stokes's theorem: we can take $ F = [0,1 ) + \omega [0,1) $,
\[\begin{split}  c_1 ( L ) &= \frac{i}{ 2 \pi } \int_{ F }  \partial_{\bar z }  \partial_{z }  \log h  \, 
d \bar z \wedge dz = 
\frac{i}{ 2 \pi } \oint_{\partial F } \partial_{ z } \log h ( z) d z \\
& = \frac{i} { 2 \pi } \int_0^1 \left( \partial_z \log h ( t ) + \omega \partial_z \log h  ( 1 + t \omega ) 
- \partial_z \log h ( \omega + t ) -  \omega \partial_z \log h ( t \omega ) \right) dt 
\end{split} .\]
Now,
\[ \begin{split} %& \partial_z \log h (\omega -t \omega ) = \partial_z \log h ( - t \omega ) - e_\omega ( -t \omega )^{-1} 
%e'_\omega ( -t \omega ) , \\
& \omega \partial_z \log h ( 1 + t \omega ) = \partial_t \left( \log h ( t \omega ) - 
\log e_1 ( t \omega )\right ) , \\
&  \partial_z \log h ( \omega + t ) = \partial_t \left(  \log h (  t ) - \log e_\omega ( t ) \right)  ,
\end{split} \]
and
\begin{equation}
\label{eq:c1L} \begin{split} c_1 ( L ) & = \frac{i} { 2 \pi } \int_0^1 \partial_t 
\left( \log e_\omega (t )  - \log e_1 ( \omega t ) \right)  dt \\
& = \frac{i} { 2 \pi } \left( \log e_\omega ( 1 ) - \log e_\omega ( 0 ) + \log e_1 ( 0 ) - \log e_1 ( \omega ) \right) , \end{split} 
\end{equation}
where we choose entire functions $ \log e_\lambda ( z ) $, which are determined up to 
an integral multiple of $ 2 \pi i $.
From \eqref{eq:CD} we see that  
$ e_\omega ( 1 ) e_1 ( 0 ) = e_{ \omega + 1} ( 0 ) = e_1 ( \omega ) e_\omega ( 0 ) $, 
and that implies (by taking logarithms) that the right hand side of \eqref{eq:c1L} is an 
integer. (We note that $ 0 $ can be replaced by any $  z \in \mathbb C $.) 

 The hermitian metric on $ L $ is called {\em strictly positive} if $ \Delta \log h < 0 $, that
is, the locally defined function $ - \log h $ is strictly subharmonic. In this case,
$ \Omega $ also defines a K\"ahler structure (of course in the very special one dimensional
case):
\[   g = - \partial_{\bar z} \partial_z \log h ( z ) |dz|^2 . \]

\end{document}